\documentclass[12pt,english]{article}
\usepackage[a4paper, portrait,
	left=0.75in,
	right=0.75in,
	top=1in,
	bottom=0.5in,
	footskip=.25in]{geometry}
\usepackage[english]{babel}
\usepackage[utf8]{inputenc}
\usepackage[T1]{fontenc}
\usepackage{helvet,etoolbox,graphicx,titlesec}
\usepackage{caption,subcaption}
\usepackage[affil-it]{authblk} 
\usepackage{etoolbox}
\usepackage{lmodern}

\usepackage{mathrsfs}
\usepackage{amsfonts}
\usepackage{amsmath} 
\usepackage{amssymb} 

\usepackage{titlesec}

\usepackage{mathrsfs}

\usepackage{graphicx}
\usepackage{dcolumn}
\usepackage{bm}

\usepackage{amsthm}
\usepackage{xspace}
\usepackage{xfrac}
\usepackage[mathcal]{euscript}
\usepackage{tikz}
\usepackage{ifthen}

\usepackage{pgfplots}
\usepgfplotslibrary{patchplots,colormaps}
\pgfplotsset{width=7cm,compat=1.18,
colormap={mycolormap}{color=(black) color=(black!20!white)}}
\usepgfplotslibrary{fillbetween} 
\usepgflibrary{shadings}
\usetikzlibrary{backgrounds}

\usepackage{epsdice}
\usepackage{twemojis}

\usepackage{pifont}

\usepackage{tcolorbox}
\tcbuselibrary{theorems,skins,breakable}
\usepackage{float}
\usepackage{colortbl}

\usepackage{natbib}

\usepackage{orcidlink}

\usepackage{url}
\usepackage{hyperref}
\usepackage{color}
\definecolor{refcolor}{RGB}{160,35,0}
\definecolor{hrefcolor}{RGB}{0,35,190}
\hypersetup{
    colorlinks,
    citecolor=refcolor,
    filecolor=refcolor,
    linkcolor=hrefcolor,
    urlcolor=hrefcolor
}

\usepackage{booktabs}
\usepackage{bookmark}
\bookmarksetup{
  numbered, 
  open,
}

\definecolor{greenPsi}{rgb}{0.0, 0.375, 0.0}
\definecolor{blueStruct}{rgb}{0.0, 0.0, 1.0}
\definecolor{redStruct}{rgb}{1.0, 0.0, 0.0}

\newcommand{\boxRule}{0.15mm}

\newcommand{\boxIndent}{15pt}

\newcommand{\remColor}{green}
\newcommand{\quoteColor}{green!50!yellow}
\newcommand{\questColor}{red}
\newcommand{\figColor}{orange}
\newcommand{\tabColor}{green}
\newcommand{\controlColor}{yellow}
\newcommand{\abstractColor}{blue!80!cyan}

\newenvironment{frameEnv}[1]
	{\begin{tcolorbox}[breakable,enhanced,toprule at break=0pt,bottomrule at break=0pt,before skip balanced=0.3cm,boxrule=\boxRule,left=0.75mm,right=0.75mm,frame hidden,borderline north = {\boxRule}{0pt}{#1!50!black}, borderline south = {\boxRule}{0pt}{#1!50!black},arc=0mm,colframe=#1!50!black,colback=#1!10,before upper={\parindent\boxIndent}]}
	{\end{tcolorbox}}

\newenvironment{frem}
	{\begin{frameEnv}{\remColor}}
	{\end{frameEnv}}
\newenvironment{fquest}
	{\begin{frameEnv}{\questColor}}
	{\end{frameEnv}}
\newenvironment{fquote}
	{\begin{frameEnv}{\quoteColor}}
	{\end{frameEnv}}

	\newenvironment{frameEnvMargin}[1]
	{\begin{tcolorbox}[breakable,enhanced,toprule at break=0pt,bottomrule at break=0pt,before skip balanced=0.3cm,boxrule=\boxRule,left=0.75mm,right=0.75mm,top=5mm,bottom=5mm,frame hidden, borderline north = {\boxRule}{0pt}{#1!50!black}, borderline south = {\boxRule}{0pt}{#1!50!black},arc=0mm,colframe=#1!50!black,colback=#1!10,before upper={\parindent\boxIndent}]}
	{\end{tcolorbox}}

\newenvironment{fabstract}
	{\begin{frameEnvMargin}{\abstractColor}\begin{abstract}}
	{\end{abstract}\end{frameEnvMargin}}

\newtheorem{theorem}{Theorem}

\theoremstyle{remark}

\numberwithin{proofStep}{theorem} 

\theoremstyle{definition}
\newtheorem{defCustom}{Definition of}
\renewcommand{\thedefCustom}{\arabic{definition}}
\makeatletter
\newcommand{\setdefCustomtag}[1]{
  \let\oldthedefCustom\thedefCustom
  \renewcommand{\thedefCustom}{#1}
  \g@addto@macro\enddefCustom{
    \global\let\thedefCustom\oldthedefCustom}
  }
\makeatother

\theoremstyle{definition}
\newtheorem{observation}{Observation}

\newtheorem{experiment}{Experiment}

\newtheorem{condition}{Condition}
\renewcommand{\thecondition}{\arabic{condition}}
\makeatletter
\newcommand{\setconditiontag}[1]{
  \let\oldthecondition\thecondition
  \renewcommand{\thecondition}{#1}
  \g@addto@macro\endcondition{
    \global\let\thecondition\oldthecondition}
  }
\makeatother

\newtheorem{proposition}{Proposition}

\newtheorem{corollary}{Corollary}
\newtheorem{definition}{Definition}

\theoremstyle{remark}

\newtheorem{remark}{Remark}
\newtheorem{example}{Example}

\newtheorem{objectType}{Type}

\newtheorem{objectTypeCustom}{Type}
\makeatletter
\newcommand{\setobjectTypeCustomtag}[1]{
  \let\oldtheobjectTypeCustom\theobjectTypeCustom
  \renewcommand{\theobjectTypeCustom}{{#1}}
  \g@addto@macro\endobjectTypeCustom{
    \global\let\theobjectTypeCustom\oldtheobjectTypeCustom}
  }
\makeatother

\newcommand{\orcid}[1]{\href{https://orcid.org/#1}{\textcolor[HTML]{A6CE39}{\aiOrcid}}}

\def\({\left(}
\def\){\right)}
\def\[{\left[}
\def\]{\right]}

\newcommand{\tn}{\textnormal}

\newcommand{\hilbert}{\mathcal{H}}

\newcommand{\mc}[1]{\mathcal{#1}}

\newcommand{\wh}[1]{\widehat{#1}}

\newcommand{\wt}[1]{\widetilde{#1}}

\newcommand{\R}{\mathbb{R}}
\newcommand{\C}{\mathbb{C}}
\newcommand{\Z}{\mathbb{Z}}

\newcommand{\N}{\mathbb{N}}

\newcommand{\abs}[1]{\left|#1\right|}

\newcommand{\de}{\operatorname{d}}

\newcommand{\etc}{\textit{etc}}

\newcommand{\stab}{\operatorname{Stab}}
\newcommand{\orb}{\operatorname{Orb}}

\newcommand{\SL}{\tn{SL}}
\newcommand{\OO}{\tn{O}}
\newcommand{\SO}{\tn{SO}}
\newcommand{\Spin}{\tn{Spin}}
\newcommand{\U}{\tn{U}}

\newcommand{\SU}{\tn{SU}}

\newcommand{\schrod}{Schr\"odinger}

\newcommand{\bra}[1]{\langle#1|}
\newcommand{\ket}[1]{|#1\rangle}
\newcommand{\braket}[2]{\langle#1|#2\rangle}

\newcommand{\rayket}[1]{\left|#1\right)}
\newcommand{\raybra}[1]{\left(#1\right|}
\newcommand{\raybraket}[2]{\left(#1|#2\right)}

\newcommand{\x}{\mathbf{x}}

\newcommand{\n}{\mathbf{n}}

\usepackage{authblk}

\newsavebox\affbox
\author{Cristi Stoica\ \orcidlink{0000-0002-2765-1562}}
\affil{Dept. of Theoretical Physics, NIPNE---HH, Bucharest, Romania.\\
Email: \textit{\color{cyan}\href{mailto:cristi.stoica@theory.nipne.ro}{cristi.stoica@theory.nipne.ro},  \href{mailto:holotronix@gmail.com}{holotronix@gmail.com}}}

\makeatletter
\newcommand*\@secondofsix[6]{#2}
\newcommand{\addtotitleformat}{%
  \@ifstar{\addtotitleformat@star}{\addtotitleformat@nostar}}
\newcommand\addtotitleformat@nostar[2]{%
  \PackageError{titlesec}{non starred form of \string\addtotitleformat\space not supported}{}}
\newcommand\addtotitleformat@star[2]{%
  \expandafter\expandafter\expandafter\expandafter
  \expandafter\expandafter\expandafter\def
  \expandafter\expandafter\expandafter\expandafter
  \expandafter\expandafter\expandafter\@currentsection@font
  \expandafter\expandafter\expandafter\expandafter
  \expandafter\expandafter\expandafter{%
    \expandafter\expandafter\expandafter\@secondofsix
       \csname ttlf@\expandafter\@gobble\string#1\endcsname}%
  \titleformat*{#1}{\@currentsection@font#2}%
}
\makeatother

\titlespacing\section{0pt}{20pt plus 6pt minus 4pt}{16pt plus 4pt minus 4pt}
\titlespacing\subsection{12pt}{12pt plus 6pt minus 4pt}{10pt plus 4pt minus 4pt}
\titlespacing\subsubsection{12pt}{12pt plus 6pt minus 4pt}{10pt plus 4pt minus 4pt}

\titleformat{\section}{\normalfont\fontsize{16}{24}\bfseries}{\thesection.}{1em}{}
\titleformat{\subsection}{\normalfont\fontsize{14}{20}\bfseries}{\thesubsection.}{1em}{}
\titleformat{\subsubsection}{\normalfont\fontsize{13}{18}\bfseries}{\thesubsubsection.}{1em}{}

\titleformat{\author}{\normalfont\fontsize{14}{20}\bfseries}{\thesection}{1em}{}

\makeatletter
\renewcommand{\thesubsection}{\arabic{section}.\arabic{subsection}}
\makeatother



\definecolor{titcolor}{RGB}{0,90,255}
\addtotitleformat*{\section}{\Large\sffamily\color{titcolor}}
\addtotitleformat*{\subsection}{\large\sffamily\color{titcolor}}
\addtotitleformat*{\subsubsection}{\large\sffamily\color{titcolor}}

\title{\color{titcolor}\textbf{Is the Wavefunction Already an Object on Space?}\footnote{Published in Stoica, O.C. \emph{Is the Wavefunction Already an Object on Space?} Symmetry 2024, 16(10), 1379.}}

\date{\small\today} 

\begin{document}

\pagestyle{headings}	
\newpage
\setcounter{page}{1}
\renewcommand{\thepage}{\arabic{page}}

\maketitle

\begin{fabstract}
Since the discovery of quantum mechanics, the fact that the wavefunction is defined on the $3\mathbf{n}$-dimensional configuration space rather than on the $3$-dimensional space has seemed uncanny to many, including Schr\"odinger, Lorentz, and Einstein. Even today, this continues to be seen as a significant issue in the foundations of quantum mechanics.
In this article, it will be shown that the wavefunction is, in fact, a genuine object on space. While this may seem surprising, the wavefunction does not possess qualitatively new features that were not previously encountered in objects known from Euclidean geometry and classical physics. The methodology used involves finding equivalent reinterpretations of the wavefunction exclusively in terms of objects from the geometry of space. The result is that we will find the wavefunction to be equivalent to geometric objects on space in the same way as was always the case in geometry and physics. This will be demonstrated to hold true from the perspective of Euclidean geometry, but also within Felix Klein's Erlangen Program, which naturally fits into the classification of quantum particles by the representations of spacetime isometries, as realized by Wigner and Bargmann, adding another layer of confirmation. These results lead to clarifications in the debates about the ontology of the wavefunction. From an empirical perspective, we already take for granted that all quantum experiments take place in space. I suggest that the reason why this works is that they can be interpreted naturally and consistently with the results presented here, showing that the wavefunction is an object on space.\end{fabstract}


\maketitle


\section{Introduction}
\label{s:intro}

In nonrelativistic quantum mechanics (NRQM), the wavefunction of $\n$ particles is defined on the $3\n$-dimensional configuration space $\R^{3\n}$,
\begin{equation}
\label{eq:wavefunction}
\begin{split}
&\psi:\R^{3\n}:=\underbrace{\R^3\times\ldots\times\R^n}_{\n\tn{ times}}\to\C \\
&\psi(\x_1,\ldots,\x_{\n})\in\C.
\end{split}
\end{equation}

While $\x_1,\ldots,\x_{\n}\in\R^3$ are points in space, at least in this representation, the wavefunction cannot be understood as defined at each point of space, but as a complex function defined on the configuration space $\R^{3\n}$.

Since the dawn of quantum mechanics (QM), this has been seen as problematic by many, including {\schrod} \citep{BacciagaluppiValentini2009SolvayConference}, Lorentz (\citealp{Przibram1967LettersWaveMechanics}, p. 44), Einstein \citep{Howard1990EinsteinWorriesQM,FineBrown1988ShakyGameEinsteinRealismQT}, Heisenberg, Bohm \citep{Bohm2004CausalityChanceModernPhysics}, etc.
For example, in a letter to {\schrod}, Lorentz wrote about how satisfied he was with {\schrod}'s wave mechanics over Heisenberg's matrix mechanics, but complained about the apparent impossibility of interpreting the wavefunction as a physical wave (\citealp{Przibram1967LettersWaveMechanics}, p. 43--44):

\begin{fquote}
If I had to choose now between your wave mechanics and the matrix mechanics, I would give the preference to the former, because of its greater intuitive clarity, so long as one only has to deal with the three coordinates $x$, $y$, $z$. If, however, there are more degrees of freedom, then I cannot interpret the waves and vibrations physically, and I must therefore decide in favor of matrix mechanics.
\end{fquote}

{\schrod} himself was deeply dissatisfied, considering the representation of the wavefunction on the configuration space ``only as a mathematical tool'' (\citealp{BacciagaluppiValentini2009SolvayConference}, p. 477):

\begin{fquote}
Of course this use of the $q$-space is to be seen only as a mathematical tool, as it is often applied also in the old mechanics; ultimately, in this version also, the process to be described is one in space and time. In truth, however, a complete unification of the two conceptions has not yet been achieved. Anything over and above the motion of a single electron could be treated so far only in the \emph{multi}-dimensional version...
\end{fquote}

He expressed more similar worries in \citep{Przibram1967LettersWaveMechanics,Przibram2011LettersWaveMechanics,Schrodinger1926QuantisierungAlsEigenwertproblem,Schrodinger1982CollectedPapersWaveMechanics}.

Similarly, Einstein wrote to Ehrenfest (\citealp{Howard1990EinsteinWorriesQM}, 28 August 1926) that

\begin{fquote}
{\schrod} is, in the beginning, very captivating. But the waves in $n$-dimensional coordinate space are indigestible.
\end{fquote}

Einstein voiced similar worries in letters to Lorentz (1 May 1926), Ehrenfest (18 June 1926), Lorentz (22 June 1926), Sommerfeld (21 August 1926), and Lorentz (16 February 1927) \citep{Howard1990EinsteinWorriesQM}, and also in \citep{FineBrown1988ShakyGameEinsteinRealismQT}.
Bohm had similar concerns {\citep{Bohm2004CausalityChanceModernPhysics}}.

Today, this is still considered to be a problem, because little progress in understanding and solving it has been made \citep{Monton2006QM3Nspace,NeyAlbert2013TheWaveFunctionEssaysOnTheMetaphysicsOfQuantumMechanics,Norsen2017FoundationsQM,Gao2017MeaningWavefunction,Maudlin2019PhilosophyofPhysicsQuantumTheory}.

But many modern physicists do not seem to worry about this, because the theory can be used successfully to explain the results of experiments, and maybe this gives us the feeling that there is no problem and no need for deeper understanding.

Arguments that we should simply embrace the fact that the wavefunction is defined on the configuration space were given in \citep{DavidAlbert1996ElementaryQuantumMetaphysics,Loewer1996HumeanSupervenience,AlbertLoewer1996TailsSchrodingerCat,Lewis2004LifeConfigurationSpace,Ney2012TheStatusOfOurOrdinaryThreeDimensionsInAQuantumUniverse,Ney2013OntologicalReductionWavefunctionOntology,North2013StructureOfQuantumWorld,Albert2019How2TeachQM}, in particular in the context of Everett's interpretation \citep{Barrett1999TheQuantumMechanicsOfMindsAndWorlds,Wallace2002WorldsInMWI,Wallace2003EverettAndStructure,BrownWallace2005BohmVsEverett,Barrett2017TypicalWorlds,SEP-Vaidman2021MWI}. Objections were raised in \citep{Monton2002WavefunctionOntology,Monton2006QM3Nspace,Maudlin2007CompletenessSupervenienceOntology,Allori2008CommonBMandGRW,Maudlin2010CanTheWorldBeOnlyWavefunction,Maudlin2013TheNatureOfTheQuantumState,Monton2013Against3NSpace,EKChen2017OurFundamentalPhysicalSpace,Emery2017AgainstRadicalQuantumOntologies,Maudlin2019PhilosophyofPhysicsQuantumTheory}. More discussions of the problem can be found in \citep{NeyAlbert2013TheWaveFunctionEssaysOnTheMetaphysicsOfQuantumMechanics}.

Let us see what various modern and contemporary researchers think about this problem. According to Bell, it is important to explain everything in terms of ``elements of reality'' associated with bounded regions of space or spacetime. For these elements of reality or properties, he uses the term ``local beables'' (to contrast them from ``observables'', which are not necessarily localized) {\citep{Bell2004TheTheoryOfLocalBeables}}:

\begin{fquote}
We will be particularly concerned with \emph{local} beables, those which (unlike, for example, the total energy) can be assigned to some bounded space-time region.
\end{fquote}

Bell used the idea of local beables to give a clear formulation of his theorem {\citep{Bell64BellTheorem}}, and to argue that the wavefunction, which he considered to be a beable, is not a local beable, resulting, upon the wavefunction collapse, in violations of local causality {\citep{Bell2004LaNouvelleCuisine}}.
Moreover, no observable of the wavefunction as it is can be a local beable, because the wavefunction can be a superposition of wavefunctions which are eigenstates corresponding to different eigenvalues of that observable.
Then, Bell argued numerous times that local beables should be in addition to the wavefunction, for example, as point-particles with definite positions, as in the pilot-wave theory (PWT), or as the points in space corresponding to the spontaneous collapse of the wavefunction in the Ghirardi--Rimini--Weber proposal (GRW), resulting in the GRW theory with the flash ontology (GRWf) \citep{Bell2004SpeakableFlashOntology}.

Bell's school of thought was born. It is based on making a case against interpreting the wavefunction as providing the complete ontology, in particular against Everett's interpretation of QM, while arguing that the point-particles in the PWT or GRW interpretation of QM with the flash ontology achieve the desired space ontology \citep{Maudlin2010CanTheWorldBeOnlyWavefunction,Maudlin2019PhilosophyofPhysicsQuantumTheory,Norsen2017FoundationsQM}.
Another sub-school denies the reality of the wavefunction, considering it to be nomological, i.e., having the same ontological status as a physical law. In PWT, the wavefunction participates in the law as expressed in the guiding equation {\citep{GoldsteinZanghi2013RealityAndRoleOfWavefunction}}, while, in the GRWf, it governs the probability distribution of the occurrence of the next spontaneous collapse \mbox{\citep{Allori2013PrimitiveOntologyAndTheStructureOfFundamentalPhysicalTheories}.}

Another argument is along the line that, since the wavefunction is an object on the configuration space, it cannot, at least not in an obvious way, result in spatial objects like chairs, tables, and tigers. By contrast, interpretations of QM like PWT and GRW theory include an additional ontology specific to the three-dimensional space, and therefore they account for the manifest image in a straightforward way, just like classical physics does.

A very good summary of the status of these problems is given by Ney in {\citep{Ney2012TheStatusOfOurOrdinaryThreeDimensionsInAQuantumUniverse}}:

\begin{fquote}
Recently, two strategies to address this question have emerged. First, Tim Maudlin, Valia Allori, and her collaborators argue that what I have just called the `most straightforward' interpretation of quantum mechanics is not the correct one. Rather, the correct interpretation of realist quantum mechanics has it describing the world as containing objects that inhabit the ordinary three-dimensional space of our manifest image. By contrast, David Albert and Barry Loewer maintain the straightforward, wavefunction ontology of quantum mechanics, but attempt to show how ordinary, three-dimensional space may in a sense be contained within the high-dimensional configuration space the wavefunction inhabits.
\end{fquote}

In the literature discussing these arguments, we can disentangle three different problems attributed to the wavefunction:
\begin{fquest}
\begin{enumerate}
	\item 
Can the wavefunction be an object in three-dimensional space?
	\item 
Can the wavefunction account for the manifest image familiar to our experience of the world?
	\item
Can the wavefunction provide a locally causal picture of the world?
\end{enumerate}
\end{fquest}

The aim of this article is only to address the first of these questions, by showing that it is an object on the three-dimensional space. I will not try to show here how the manifest image, containing spatial objects like chairs and tables and tigers, emerges from the wavefunction.
This has been discussed extensively and convincingly in the literature, notably in \citep{Wallace2012TheEmergentMultiverseQuantumTheoryEverettInterpretation,Vaidman2022WaveFunctionRealismAnd3Dimensions}.
And, while the fact that interactions take place in space plays a central role in \citep{DavidAlbert1996ElementaryQuantumMetaphysics,Vaidman2022WaveFunctionRealismAnd3Dimensions}, the existence of a three-dimensional manifest image does not even require the wavefunction to be an object in space; it only requires it to be able to produce the appearance of such objects. Therefore, the fact proposed here that the wavefunction already is an object on space is not even necessary to solve the problem of the manifest image. I hope that this result, proved here in a classical geometric sense, will increase confidence in the fact that the wavefunction can account for the physical reality.

I will not attempt to give a locally causal picture of the world in this article either. This discussion is neutral with respect to wavefunction collapse, being consistent both with the many-worlds interpretation (MWI) and with the wavefunction collapse, and, in fact, with any interpretation of QM, since all of them require the existence of the wavefunction {\citep{PBR2012RealityOfPsi}}. However, as long as no collapse takes place, the interactions encoded in the wavefunction are local, if we understand it in a particular way as a field with infinitely many components in a fiber bundle, as explained in {\citep{Stoica2019RepresentationOfWavefunctionOn3D}}. If the wavefunction collapse is assumed to take place, it breaks the locality. In fact, some researchers think that the MWI can avoid violations of local causality without appealing to the monstrous construction from {\citep{Stoica2019RepresentationOfWavefunctionOn3D}}. According to Vaidman {\citep{SEP-Vaidman2021MWI}},

\begin{fquote}
The MWI does not have action at a distance. The most celebrated example of nonlocality of quantum mechanics given by Bell's theorem in the context of the Einstein-Podolsky-Rosen argument cannot get off the ground in the framework of the MWI because it requires a single outcome of a quantum experiment, see the discussion in Bacciagaluppi 2002, Brown and Timpson 2016. Although the MWI removes the most bothersome aspect of nonlocality, action at a distance, the other aspect of quantum nonlocality, the nonseparability of remote objects manifested in entanglement, is still there. A ``world'' is a nonlocal concept. This explains why we observe nonlocal correlations in a particular world.
\end{fquote}

But the subject of this article is more specific, dealing only with the status of the wavefunction as an object on space, as understood in geometry or in classical physics.

A simple way to understand the wavefunction as an object in space is to interpret $\psi(\x_1,\ldots,\x_{\n})$ as a multifield, which is similar to a field but depends on multiple positions in space \citep{Forrest1988QuantumMetaphysics,Belot2012QuantumStatesPrimitiveOntologists}.
However, it can also be understood as an infinite number of classical fields on the regular three-dimensional space, interacting locally, in at least two different reconstructions. One is obtained by brute force \citep{Stoica2019RepresentationOfWavefunctionOn3D} and another one emerges naturally in the {\schrod} wavefunctional formulation of quantum field theory, at least in the case of bosonic fields \citep{Stoica2023TheRelationWavefunction3DSpaceMWILocalBeablesProbabilities}.

The brute force construction given in \citep{Stoica2019RepresentationOfWavefunctionOn3D} is performed in several steps. Recall that the total wavefunction $\psi(\x_1,\ldots,\x_{\n})$ is a linear superposition of products of wavefunctions of the individual particles, which are defined on the three-dimensional space. The terms in the superposition are of the form $\psi(\x_1)_1\cdot\ldots\cdot\psi_{\n}(\x_{\n})$. The first step of the construction is to collect these individual wavefunctions as the components of a field with infinitely many components, $\(\psi_1(\x_1),\ldots,\psi_{\n}(\x_{\n})\)$. But this contains too much information, because $\(c\psi_1(\x_1),\ldots,\psi_{\n}(\x_{\n})\)=\(\psi_1(\x_1),\ldots,c\psi_{\n}(\x_{\n})\)=c\(\psi_1(\x_1),\ldots,\psi_{\n}(\x_{\n})\)$, where $c\in\C$. This is resolved by introducing a global gauge symmetry. The next step is to build fields with even more components, so that they represent linear combinations of fields as in the first step and so that the wavefunctions belong to a fixed basis. Again, we need an even larger global gauge symmetry group. By promoting this global gauge symmetry to a local gauge symmetry and using a connection, the fields are made locally separable. This construction seems very artificial and monstrous, despite it being backed up by a relation between fiber bundles and tensor products of spaces of functions. There seems to be no significant conceptual gain in using it, except that it is a proof of the concept that all the information contained in the wavefunction can be seen as a very complicated field on space.
Also, calculations are much easier when using the wavefunction on the configuration space.

A more natural construction, based on the wavefunctional formulation of quantum field theory for bosonic fields, was presented in \citep{Stoica2023TheRelationWavefunction3DSpaceMWILocalBeablesProbabilities}. In this construction, a wavefunctional $\Phi$ on the configuration space of classical fields can be understood as a field with infinitely many components, with each of these components being a classical field configuration $\phi$ along with $\Phi[\phi]$, the value of the wavefunctional $\Phi$ for the field configuration $\phi$. However, it is unclear at this point how this would work for fermionic fields (Remark \ref{rem:thm:qf_space:fermions}).

Both these representations constitute a proof of concept, helping us to understand that the wavefunction may be just a simpler way to represent something that secretly is an object on space. It would be useful to also have a way of understanding this that is more similar to the way of understanding geometric objects in space
, at least for mathematicians, and with minimal reinterpretations of the wavefunction on the configuration space, which is an extremely useful concept
.
In this article, I will show that the wavefunction already is an object precisely like the objects in Euclidean geometry or in classical physics
.

The article is organized as follows: Section \ref{s:types} makes an inventory of the types of objects on space in geometry and in classical physics. Section \ref{s:wf_space} contains a proof that the wavefunction is such an object and does not have qualitatively new properties compared to other objects on space.
In Section \ref{s:klein}, it is shown how this is true also in the precise sense of Klein's Erlangen Program. Section \ref{s:wigner} discusses this result in relation to the classification of quantum particles by the representations of the spacetime isometries realized by Wigner and Bargmann. Section \ref{s:empirical} explains that the already existing empirical data are consistent with these conclusions, and they are naturally interpreted as about objects on space. The final section, Section \ref{s:discussion}, discusses these results in the context of the idea, raised by critics, that the wavefunction is sufficient in quantum mechanics.

\section{Euclidean Geometry Explanation}
\label{s:explanation_Euclidean}

This section contains an intuitive explanation of the fact that the wavefunction already is an object on space, based on Euclidean geometry. Symmetry groups, discussed in Section \ref{s:explanation_symmetry}, will deepen this insight.

\subsection{Types of Classical Objects on Space}
\label{s:types}

Based on Euclidean geometry and classical physics, let us see what kinds of objects count as objects on space.

\begin{frem}
\begin{objectType}[Subsets of space]
\label{type:space_subset}
The most commonly used objects in space are subsets of space. 
\end{objectType}
\end{frem}

\begin{example}[of Type \ref{type:space_subset}]
\label{ex:type:space_subset}
Lines, segments, and circles are subsets of the Euclidean plane $\R^2$. Technically, points are elements of the plane, not subsets. But since statements of the form ``let the point $A$ be the intersection of the lines $b$ and $c$'', meaning in fact that $\{A\}=a\cap b$, are usual, we will include both these usages in this same type.

In classical mechanics, particles are represented as points, and solid objects as subsets of space.
\end{example}

\begin{frem}
\begin{objectType}[Composite objects]
\label{type:space_collection_subsets}
A geometric object may be composite, i.e., it may consist of a collection of subsets of the geometric space without being identified with the union of those subsets.
\end{objectType}
\end{frem}

\begin{example}[of Type \ref{type:space_collection_subsets}]
\label{ex:type:space_collection_subsets}
Consider a triangle $ABC$ in the Euclidean plane. 
Let $M\in[BC]$ be the middle of the segment $[BC]$, and let $N\in[BC]$ so that the half-line $[AN$ is the bisector of the angle $\angle{CAB}$. Then, the half-lines $[AM$ and $[AN$ are distinct, even though, in the case when $|AB|=|AC|$, they coincide as subsets of the plane. Therefore, they should be considered as distinct objects composing the figure, and considering their union instead would miss this fact.

Even though the edges $[AB]$ and $[AC]$ of the triangle $ABC$ are subsets of the plane that share the point $A$, they are distinct objects. If the triangle is degenerate so that $B\in[AC]$, $[AB]\subset[AC]$, and it is even possible that $B=C$, in which case, $[AB]=[AC]$ as sets; however, theorems that apply to general triangles remain valid only if we consider the edges $[AB]$ and $[AC]$ as distinct objects.

Such extreme and degenerate situations occur, for instance, in problems of geometric locus, where it is essential to consider the distinct points as remaining distinct even when they happen to coincide. For example, given a segment $[BC]$ in the plane, the geometric locus of the point $A$, for which the angle $\angle{CAB}$ is right, is the circle with $[BC]$ as that diameter. But if we exclude the cases when $A=B$ and $A=C$, the locus is no longer a full circle.

Classical mechanisms can be composed of rigid bodies connected by gears, belts, cams, linkages, \etc. From a geometric point of view, they are Type \ref{type:space_collection_subsets} objects too.
Classical waves and fields can occupy common regions in space and interfere, yet they remain distinct and continue to propagate in a way that is unaffected. In particular, gravity is not affected by the presence of electromagnetic waves.
\end{example}

The composite geometric objects from Example \ref{ex:type:space_collection_subsets} are very common in geometry and classical physics. Identifying composite objects with the union of the subsets would collapse Type \ref{type:space_collection_subsets} into Type \ref{type:space_subset}. This would reduce the generality of the theorems in geometry and classical physics, and they would be broken into many very particular theorems with a reduced range of applicability, undermining universality. 
Often, geometric reasoning starts with extreme or degenerate cases, with the general solution resulting from continuity.
But Type \ref{type:space_collection_subsets} objects are not merely useful; they capture general geometric facts.

\begin{frem}
\begin{objectType}[Objects with additional properties]
\label{type:space_additional_properties}
If we associate to a geometric object an element $s$ from a set $S$, the result is still a geometric object. This also applies to the components of an object of Type \ref{type:space_collection_subsets}.
\end{objectType}
\end{frem}

\begin{example}[of Type \ref{type:space_additional_properties}]
\label{ex:type:space_additional_properties}
Two lines are congruent, but they are distinct objects, so labeling them helps us track their identity. However, it is not merely a matter of utility; labeling them reflects their distinct identities and meanings, like the half-line and the bisector from Example \ref{ex:type:space_collection_subsets}.

In classical mechanics, an object with a mass and an electric charge is still an object on space. A composite object can have a total mass and a total charge, but, at the same time, its constituent parts have their own masses and charges. Therefore, labeling geometric objects is not just a convention; it reflects real properties. In physics, it is arguable whether these properties are geometric in nature, but they are nevertheless ``on space''.
\end{example}

Before continuing, let us remember some notions relating to fiber bundles.
In classical physics, but mostly in differential geometry, we often deal with fields, for example, scalar, vector, and tensor fields.
Initially, mathematicians and physicists used to think of fields as functions on a manifold $M$ with values in a set $V$ which can be the field of real or complex numbers, a vector space, etc. A field was therefore defined as a function $f:M\to V$, so that $f(x)\in V$ for all $x\in M$. However, as we understood more, we realized that, in many cases, the set $V$ should be different for each $x\in M$, that is, $f(x)\in V_x$, where $V_x$ are all identical copies of the same set $V$. An example is when we deal with tangent vectors. For example, if $M$ is a curve, a vector that is tangent to the curve at a point $x$ belongs to a vector space with the origin at $x$. Therefore, there will be as many vector spaces as there are points in curve $M$
.
If we collect together all these vector spaces, we obtain a structure that looks, in each neighborhood $(a,b)$ of $x$, like the product $(a,b)\times V$. However, if the curve is closed, when gluing together the spaces of the form $(a,b)\times V$, we may obtain different results. For example, a vector field on $M$ may twist as we move around $M$. If it twists only half of a full rotation, we obtain a M\"obius strip.
Therefore, the notion of fiber bundle was discovered. In this example, a fiber bundle consists of the total space $\mc{B}$, which is obtained by gluing together spaces of the form $(a,b)\times V$ for open intervals $(a,b)$, that, together, cover the entire curve, and a projection map $\pi$, so that, if we are handed a vector $v_x$ from the vector space $V_x$, $\pi(v_x)=x$
.

In general, a fiber bundle consists of a manifold $\mc{B}$, a base space $M$, and a surjective function $\pi:\mc{B}\to M$, the projection map, so that $\pi^{-1}(x)$ is ``the same'' for all $x\in M$. The set $\pi^{-1}(x)$ is the fiber above the point $x$.
A field is a section of the bundle, which is just a map that associates to each point $x$ an element of the fiber over $x$.
We denote the fiber bundle by $\mc{B}\stackrel{\pi}{\mapsto} M$.

In the case of a curve $M$ in the plane $\R^2$, the vectors in the plane with an origin at a point $x$ of $M$ form a two-dimensional vector space $\R^2_x$, the total space $\mc{B}=M\times\R^2$, the projection $\pi$ takes a vector $(x,v_x)\in M\times\R^2$ into the point $x$, and the fiber bundle is $M\times\R^2\stackrel{\pi}{\mapsto} M$.

To differentiate a field of the bundle, we need a way to compare the value of the field at $x+dx$ with its value at $x$. This is provided by a connection, which specifies what happens to the vector or the tensor when we translate it along a curve. The connection can then be used to define the curvature of the bundle at a point $x$.
An example of a connection in physics is the Levi--Civita connection, necessary to derivate vector and tensor fields in general relativity and in Riemannian geometry.
Another example is the connection appearing in the gauge theory of electromagnetism, which uses a fiber bundle with the base manifold being the space or spacetime, and the typical fiber being the group $\U(1)$. The electromagnetic field turns out to be the curvature of the connection. When expressed in coordinates, the connection becomes the potential four-vector of the electromagnetic field.

For the purpose of this article, we do not need to go into more detail regarding these abstract notions, but it is important to understand that they are considered the object of local differential geometry; they are considered to be local objects {\citep{Guggenheimer1977DifferentialGeometry}}.

\begin{remark}
\label{rem:type:space_additional_properties}
One may argue that, except for in tracking the identity of geometric objects, additional properties are not genuine properties of Euclidean geometry because they are neither Type \ref{type:space_subset} nor Type \ref{type:space_collection_subsets}, which is true. Type \ref{type:space_additional_properties} includes additional elements that cannot be reduced to points in space. But this was never a reason to consider that properties like mass and electric charge in classical physics are not on space.

Consider the charge density. It depends on position, so it is a scalar field $\rho:\R^3\to \R$.
In general, classical physics contains among its objects various scalar, vector, and tensor fields.
This justifies the geometric interpretation of these fields as sections of various fiber bundles over the space $\R^3$.
For example, scalar fields are sections of the trivial bundle $\R^3\times\R\mapsto\R^3$, vector fields are sections of $\R^3\times\R^d\mapsto\R^3$, \etc.
Here, ``$\mapsto$'' represents the projection on the first component of the Cartesian product, and ``sections'' are continuous subsets of $\R^3\times\R$ or $\R^3\times\R^d$ whose points are put into a one-to-one correspondence with the points of the space $\R^3$ by the projection. Hence, the notion of the section captures the idea of a scalar or vector field in a geometric way.
Complex fields are also possible in classical physics, and they can be trivially interpreted as real fields since any $d$-dimensional complex vector space $\C^d$ has an underlying real $2d$-dimensional vector space $\R^{2d}$. Moreover, it turns out that some of these fields are more appropriately interpreted as special geometric objects defined on these bundles. Notably, the electromagnetic potential is understood as the components of a connection of the bundle $\R^{3+1}\times \U(1)\mapsto\R^{3+1}$, where $\R^{3+1}$ is the Minkowski spacetime, while the electromagnetic field is its curvature. This is the geometric formulation of gauge theory, and it generalizes to non-Abelian groups corresponding to the electroweak and strong interactions, resulting in the (classical) Yang--Mills theory \citep{YangMills1954ConservationOfIsotopicSpinAndIsotopicGaugeInvariance}.

We see that the additional properties appearing in classical physics, while not geometric objects inherent to the $3$-dimensional space, can be interpreted as objects of Types \ref{type:space_subset} and \ref{type:space_collection_subsets} of extended spaces---the fiber bundles.
While fiber bundles extend the geometry of space or spacetime, the projection ``$\mapsto$'' makes their sections, connections, and curvatures genuine geometric objects on the base space, that is, on the Euclidean space $\R^3$, on the Minkowski spacetime $\R^{3+1}$, or on the four-dimensional curved spacetime of general relativity.
\end{remark}

The discussion from Remark \ref{rem:type:space_additional_properties} justifies the following geometric version (and generalization) of Type \ref{type:space_additional_properties} objects.

\begin{frem}
\setobjectTypeCustomtag{\ref*{type:space_additional_properties}$^\prime$}
\begin{objectTypeCustom}[Geometrization of additional properties]
\label{type:space_additional_properties_geometrization}
Geometric objects of Types \ref{type:space_subset} and \ref{type:space_collection_subsets} of a fiber bundle over a manifold $M$ (which can be the space or spacetime) are objects on the base space $M$.
\end{objectTypeCustom}
\end{frem}

\begin{example}[of Type \ref{type:space_additional_properties_geometrization}]
\label{ex:type:space_additional_properties_geometrization}
All examples from Remark \ref{rem:type:space_additional_properties}.
Note that objects of Type \ref{type:space_additional_properties_geometrization} do not have to be sections of the bundle. For example, for a charged classical point-particle, the charge is localized at a point; therefore, if we interpret it as a section of the bundle $\R^3\times \R\mapsto\R^3$, it is a discontinuous section.
In general, geometric objects of Types \ref{type:space_subset} and \ref{type:space_collection_subsets} of a fiber bundle over space are not sections, although they can be interpreted as discontinuous sections or collections of such sections.
\end{example}

\begin{remark}
\label{rem:type:reduction}
Any Type \ref{type:space_subset} object can be understood as a Type \ref{type:space_collection_subsets} object, where the collection of Type \ref{type:space_subset} objects has only one element.
Any Type \ref{type:space_collection_subsets} object can be understood as a Type \ref{type:space_additional_properties} object, where the additional property is the same for all elements of the collection defining the Type \ref{type:space_collection_subsets} object. Any Type \ref{type:space_collection_subsets} object can also be understood as a Type \ref{type:space_additional_properties_geometrization} object, where the fiber bundle is $\R^3\times L\mapsto\R^3$, where $L$ has one element per set of points composing the Type \ref{type:space_collection_subsets} object.
\end{remark}

Throughout this article, I refer to Type \ref{type:space_subset} and \ref{type:space_collection_subsets} objects as objects in space, and to Type \ref{type:space_additional_properties} and \ref{type:space_additional_properties_geometrization} objects as objects on space, to distinguish that the former consist of subsets of space and the latter consist of Type \ref{type:space_subset} and \ref{type:space_collection_subsets} objects with additional nongeometric properties, or geometric objects in a fiber bundle over space.

\subsection{Wavefunction on Space}
\label{s:wf_space}

We are ready to prove that the wavefunction already is an object on space.

\begin{frem}
\begin{theorem}
\label{thm:wf_space}
The wavefunction in NRQM is an object on space of the classical types described in Section \ref{s:types}.
\end{theorem}
\end{frem}
\begin{proof}
Any configuration of $\n$ points in space $(\x_1,\ldots,\x_{\n})$ is an object of Type \ref{type:space_subset}. The collection of all such $\n$-tuples is an object of Type \ref{type:space_collection_subsets}. However, to specify the wavefunction $\psi$ from Equation \eqref{eq:wavefunction}, each of these $\n$-tuples has to be associated with a complex number $\psi(\x_1,\ldots,\x_{\n})\in\C$. Therefore, the wavefunction $\psi$ is an object of Type \ref{type:space_additional_properties}.
More precisely, the wavefunction is an object of Type \ref{type:space_collection_subsets} in the fiber bundle $\R^3\times\C\mapsto\R^3$. To see this, we can interpret geometrically the pair $\((\x_1,\ldots,\x_{\n}),c\)$, where $c=\psi(\x_1,\ldots,\x_{\n})\in\C$, as a configuration of $\n$ points in $\R^3\times\C$,  $\((\x_1,c),\ldots,(\x_{\n},c)\)\subset\R^3\times\C$. This is a Type \ref{type:space_subset} object in $\R^3\times\C$.
The collection of such objects corresponding to all possible choices of $\n$ points in $\R^3$ is therefore an object of Type \ref{type:space_collection_subsets} in $\R^3\times\C$.
This makes it an object of Type \ref{type:space_additional_properties_geometrization} on the $3$-dimensional space $\R^3$.

Another way to see this is the following: we can interpret a pair $\((\x_1,\ldots,\x_{\n}),c\)$ as a function $f_{\{\x_1,\ldots,\x_{\n}\}}:\R^3\to \{0,1\}\times\C$,
\begin{equation}
\label{eq:type2-section}
f_{\{\x_1,\ldots,\x_{\n}\}}(\x):=\(\chi_{\{\x_1,\ldots,\x_{\n}\}},c\),
\end{equation}
where $\chi_M$ is the characteristic function of the set $M$, that is, $\chi(\x)=1$ if $\x\in M$ and $\chi(\x)=0$ otherwise. Therefore, $f_{\{\x_1,\ldots,\x_{\n}\}}$ is a section of a bundle $\R^3\times\(\{0,1\}\times\C\)\mapsto\R^3$.
The collection of such sections, corresponding to all possible choices of points in $\R^3$, is therefore a section $\mc{F}$ in a bundle with a very large fiber over $\R^3$, obtained as the bundle direct sum $\bigoplus_{\n\in\N}\bigoplus_{\(\x_1,\ldots,\x_{\n}\)\in\(\R^3\)^\n} \R^3\times\(\{0,1\}\times\C\)\mapsto\R^3$, namely, $\mc{F}$ has as components all functions $f_{\{\x_1,\ldots,\x_{\n}\}}$ for all $\(\x_1,\ldots,\x_{\n}\)\in\(\R^3\)^k$ for all $\n\in\N$.
Therefore, the wavefunction can be interpreted as well as a section with infinitely many components over the $3$-dimensional space $\R^3$. 
\end{proof}

\begin{remark}
\label{rem:f-are-rigid}
Note that the restriction of the function $f_{\{\x_1,\ldots,\x_{\n}\}(\x)}$ from Equation \eqref{eq:type2-section} to a subset of points $\{\x'_1,\ldots,\x'_{k}\}\subsetneq\{\x_1,\ldots,\x_{\n}\}$ is not the function $f_{\{\x'_1,\ldots,\x'_{k}\}}$ corresponding to the subset $\{\x'_1,\ldots,\x'_{k}\}$. These functions are rigid objects, just like rigid bodies in classical physics, or like Type \ref{type:space_collection_subsets} objects in Euclidean geometry.
One may object that the component $c$ of $f_{\{\x_1,\ldots,\x_{\n}\}}$, which is $\psi\(\x_1,\ldots,\x_{\n}\)$, is not localized anywhere in space since it depends on the configuration of points $\{\x_1,\ldots,\x_{\n}\}$. But consider, for example, a classical plane wave. Can we say that its wavelength is localized in some point in space? Its wavelength depends on the values of the plane wave at all points in space, or at least an extended region of space, and yet nobody would doubt the fact that the plane wave is a classical object.

However, it is possible to have a representation of the wavefunction on space as a section in a very large fiber bundle so that the section restricts in this way, to obtain local separability. This can be achieved by using a very large local gauge group on that very large bundle, as shown in \citep{Stoica2019RepresentationOfWavefunctionOn3D}.
\end{remark}

\begin{remark}
\label{rem:no-qualitative-difference}
Theorem \ref{thm:wf_space} shows that the relation between the wavefunction and the $3$-dimensional space is not qualitatively different compared to the relation between objects in Euclidean geometry or classical physics and space. Quantitatively speaking, the wavefunction is in some sense an ``extreme'' case of Type \ref{type:space_additional_properties} because, as a collection of configurations of points in space, it includes all possible tuples of points.
\end{remark}

\begin{remark}
\label{rem:multifield}
Theorem \ref{thm:wf_space} can be interpreted as providing additional grounding on space for the multifield interpretation of the wavefunction \citep{Forrest1988QuantumMetaphysics,Belot2012QuantumStatesPrimitiveOntologists}.
However, there is no need to adopt multifields once we realize that the wavefunction as it is already is an object on space.
\end{remark}

\subsection{Quantum Fields on Space}
\label{s:qf_space}

What about quantum fields? After all, the correct quantum theory is not NRQM but quantum field theory (QFT).
A simple way to show that quantum fields can be as well understood as objects on space follows from Theorem \ref{thm:wf_space}, as we will see now. But it is also possible to show it directly for fields in the wavefunctional representation, as we will see in Theorem \ref{thm:qf_space}.

\begin{frem}
\begin{corollary}
\label{thm:wf_space_qft}
The result of Theorem \ref{thm:wf_space} extends to quantum field theory as well.
\end{corollary}
\end{frem}
\begin{proof}
We will use the Fock representation of quantum states. Quantum fields are usually represented as operator-valued distributions, expressed in terms of creation and annihilation operators $\hat{a}^\dagger_D$ and $\hat{a}_D$. These, in turn, are subject to commutation (in the case of bosons) or anticommutation relations (in the case of fermions). The subscript index $D$ represents the particle type, and it may include spin degrees of freedom and internal degrees of freedom (like color and hypercharge). Let $\mc{D}$ be the set of all types of particles. We do not need to detail all possible types of Standard Model particles, as they all can be expressed by creation and annihilation operators for various types of particles from $\mc{D}$.

By repeatedly applying various operators $\hat{a}^\dagger_D$ and $\hat{a}_D$ and linear combinations to the vacuum state $\ket{0}$, any Fock state can be constructed. Since we are interested in the position basis, we apply creation operators of the position eigenstates,
\begin{equation}
\label{eq:Fock}
\ket{\x_1,\ldots,\x_\n}_D := \hat{a}^\dagger_{D_1}(\x_1)\ldots\hat{a}^\dagger_{D_\n}(\x_\n)\ket{0},
\end{equation}
where $D=\(D_1,\ldots,D_\n\)$, obtaining the Fock representation in the position basis.
Then, any quantum state $\ket{\psi}$ can be expressed as a superposition of $\n$-particle wavefunctions for all possible positive integer values of $\n$,
\begin{equation}
\label{eq:psi_Fock}
\ket{\psi} = \sum_{\n\in\N}\sum_{D\in\mc{D}^\n}\int_{(\x_1,\ldots,\x_n)\in\R^{3\n}}
\psi(\x_1,\ldots,\x_\n,D)\ket{\x_1,\ldots,\x_\n}_D\de \x_1\ldots\de\x_\n.
\end{equation}

Hence, quantum states from quantum field theory are just superpositions of $\n$-particle states, already known from Theorem \ref{thm:wf_space} to be objects on space. In the general case, from Equation \eqref{eq:psi_Fock}, they are collections of configurations of $\n$ points in the fiber bundle $\R^3\times \C\times\mc{D}$. Therefore, the infinite collection of wavefunctions from Equation \eqref{eq:psi_Fock} forms an object of Type \ref{type:space_collection_subsets} in $\R^3\times\C\times\mc{D}$, and therefore an object of Type \ref{type:space_additional_properties_geometrization} on the $3$-dimensional space $\R^3$.
\end{proof}

In the standard quantization scheme used in NRQM, one starts from classical point-particle configurations and obtains wavefunctions defined on that configuration space.
The basis of the resulting Hilbert space corresponds to points in this configuration space so they can be labeled by the positions $\ket{\x_1,\ldots,\x_\n}$.
To obtain field quantization, the resulting wavefunctions are then again quantized (the second quantization).

However, we can apply the same quantization idea directly to classical fields \citep{Hatfield2018QuantumFieldTheoryOfPointParticlesAndStrings}, leading to another way to understand quantum fields as objects on space, at least for bosonic fields.
\begin{frem}
\begin{theorem}
\label{thm:qf_space}
Quantum field configurations in the wavefunctional representation are, in the case of bosonic fields, objects on space.
\end{theorem}
\end{frem}
\begin{proof}
A classical field $\varphi$ configuration is defined on the three-dimensional space, and valued either in $\C$ or in some vector space $\mathbb{V}$, whose dimension depends on the spin and internal properties of the field.
The values of the field at different points are required to satisfy canonical commutation or anticommutation relations, according to the spin, but this does not impact this proof.
Let $\mathcal{C}(\R^3,\mathbb{V})$ be the classical field configuration space of classical bosonic fields defined on space and valued in $\mathbb{V}$.
In the wavefunctional quantization \citep{Hatfield2018QuantumFieldTheoryOfPointParticlesAndStrings}, the {\schrod} wavefunctional, or its relativistic versions, is defined on $\mathcal{C}(\R^3,\mathbb{V})$.
The {\schrod} wavefunctional $\Psi:\mathcal{C}(\R^3,\mathbb{V})\to\C$ depends on the classical field configuration $\varphi$, and associates a complex number $\Psi(\varphi)$ to each $\varphi$. The basis of the Hilbert space of wavefunctionals is in one-to-one correspondence with the classical field configurations, so we will label them $\ket{\varphi}$. The state vector has the form
\begin{equation}
\label{eq:wavefunctional}
\ket{\Psi}=\int_{\mathcal{C}(\R^3,\mathbb{V})}\Psi(\varphi)\ket{\varphi}\mathcal{D}\varphi,
\end{equation}
where $\mathcal{D}\varphi$ is the measure on $\mathcal{C}(\R^3,\mathbb{V})$.

The classical field $\varphi$ is a section in a bundle $\R^3\times\mathbb{V}\mapsto\R^3$. The wavefunctional \eqref{eq:wavefunctional} associates to each section of the bundle a complex number. The sections of $\R^3\times\mathbb{V}\mapsto\R^3$ are Type \ref{type:space_subset} objects in the bundle's total space $\R^3\times\mathbb{V}$.
A section with a label from $\C$ is a Type \ref{type:space_subset} object in $\R^3\times\mathbb{V}\times\C$, which is the total space of a bundle $\R^3\times\mathbb{V}\times\C\mapsto\R^3$. The wavefunctional \eqref{eq:wavefunctional} is thus an object of Type \ref{type:space_collection_subsets} in $\R^3\times\mathbb{V}\times\C\mapsto\R^3$, and, by Remark \ref{rem:type:space_additional_properties}, an object of Type \ref{type:space_additional_properties_geometrization} on the $3$-dimensional space $\R^3$.
\end{proof}

\begin{remark}
\label{rem:thm:qf_space:fermions}
For fermionic fields, there are important differences, preventing a straightforward interpretation like the one from Theorem \ref{thm:qf_space} for the case of the bosons \citep{Jackiw1988AnalysisInfDimManifoldsSchrodingerRepresentationForQuantizedFields}.
Technically, this result applies to fermionic fields as well, but it is difficult to interpret the wavefunctional in terms of fields localized in space. The reason is that fermionic field creation and annihilation operators localized at spacelike separated positions anticommute. The corresponding bosonic operators commute, which makes the localization of the field operators possible, but the anticommutativity of the fermionic creation and annihilation operators requires that the components of the classical fermionic fields are ``anticommuting numbers'', called Grassmann numbers \citep{Jackiw1988AnalysisInfDimManifoldsSchrodingerRepresentationForQuantizedFields,Hatfield2018QuantumFieldTheoryOfPointParticlesAndStrings}. Because the field operators for spacelike separated points anticommute, the fermionic fields cannot be localized {\citep{FloreaniniJackiw1988FunctionalRepresentationForFermionicQuantumFields}}:

\begin{fquote}
...a particle state is localized in the bosonic functional space, while there is no concept of localization in the Grassmann space.
\end{fquote}

However, since the observables are quadratic in the spinor fields, they can be local, and local measurements at spacelike separated points are compatible. This nonlocalization affects only the fermionic field operators and the creation and annihilation operators and not the observables, so it does not manifest in experiments, and everything is found to be localized. However, this nonlocalizability seems to prevent a straightforward method to assign an ontology on space or spacetime to the fermionic fields.

Despite this problem, the representation given in {\citep{Stoica2019RepresentationOfWavefunctionOn3D}}, while not given in terms of wavefunctionals, works for both bosons and fermions because it does not make use of anticommuting numbers. It is based on the Fock representation in terms of tensor products of wavefunctions on space
. Since the wavefunctional and the Fock formulations are equivalent {\citep{Hatfield2018QuantumFieldTheoryOfPointParticlesAndStrings}}, this suggests that there must be a way to make sense of fermionic fields as objects on space for the wavefunctional formulation as well, but how exactly this can be achieved remains an open question.
\end{remark}

\begin{remark}
\label{rem:born-rule}
Theorem {\ref{thm:qf_space}} can be applied to give an interpretation of the Born rule in terms of classical probability distributions of coexisting classical worlds.
The first impression of someone freshly acquainted with the probabilities in QM is that they apply to the state of the observed system. But there is another way to look at this problem: by considering the state of the entire universe. This picture is not something that was introduced only after the discovery of quantum mechanics, since it was used in classical statistical mechanics too.
In classical statistical mechanics, macrostates, which are equivalence classes of macroscopically indistinguishable (micro)states, are regions of the phase space. This observation can be used to provide an interpretation of QM ``without observers'', in which the wavefunction collapse is triggered by the propagation of the Wigner function (representing the state in the phase space formulation of QM) across a region of the phase space corresponding to more macrostates, so that the wavefunction collapses to a single macrostate, with a probability postulated to be the one given by the Born rule {\citep{Stoica2020StandardQuantumMechanicsWithoutObservers}}.
However, here, I want to talk about probabilities in the MWI.
In the Hilbert space formulations of QM, the macrostates correspond to subspaces forming a direct sum decomposition of the Hilbert space. These macrostates determine the branches of the wavefunction in the MWI. But since they represent macrostates, which appear to us as classical, understanding the relation between the wavefunction and the $3$-dimensional space is important. The existence of local observables that correspond to macroscopically observable properties localized in the $3$-dimensional space becomes relevant, since the macrostates can then be understood as subspaces of the Hilbert space consisting of states with the same macroscopic classical properties. These macroscopic observables play a special role as local beables, but they are insufficient to identify a unique microstate. A completion of the commuting set of macroscopic classical observables is necessary. This completion can be achieved by using the wavefunctional formulation of QFT, and the representation from Theorem {\ref{thm:qf_space}} in terms of classical field configurations. What is interesting is that, by doing this, we achieve not only an understanding of the quantum wavefunctional in terms of classical fields on space, but also an interpretation of the probabilities in QM. If the coefficients in the linear combination from Equation {\eqref{eq:wavefunctional}} are real numbers, the wavefunctional is equivalent to a measure on the configuration space of a classical field configuration. In MWI, this can be used to give a justification for the Born rule, by reinterpreting the wavefunctional (as in Theorem {\ref{thm:qf_space}}) as a probability distribution of coexisting classical worlds. The probability is obtained as the self-location probability of the observer in one of the classical worlds. Since the probability distribution is spread differently across each macrostate, this gives a straightforward interpretation of the Born rule very similar to the probabilities in classical statistical mechanics. The main difference is that the probabilities do not come from the observer's lack of knowledge about which is the actual state of the classical world, but from the lack of knowledge about in which of the possible classical worlds she is located. In this approach, there is no independent evolution of the classical worlds in which they preserve their identity because, in fact, the dynamics are given by the {\schrod} equation for the wavefunctional, which propagates over the configuration space of classical field configurations. This cannot be interpreted as independent evolutions of the individual classical worlds, but only as a collective evolution of the distribution of classical worlds. Therefore, the classical worlds can be used only to yield probabilities, as they do in classical statistical mechanics. But what about the fact that the wavefunctional is not real but complex? Since the classical field configurations include electromagnetic gauge fields with the symmetry group $\U(1)$, there is a correspondence between the choices of the global $\U(1)$ gauge and the complex phase multiplying the basis vectors corresponding to the classical field configuration briefly described in Theorem {\ref{thm:qf_space}} so that the complex wavefunctional can be interpreted as a probability distribution of coexisting classical gauged field configurations. The global $\U(1)$ gauge, just like the global phase, is not observable, but we are dealing with a linear combination of such states, as in Equation {\eqref{eq:wavefunctional}}, and then these gauges/phases associated with the classical states matter for the collective evolution of the classical worlds composing the wavefunctional. Due to the fact that the decomposition of the Hilbert space as a direct sum of subspaces representing macrostates is the same at all times, regardless of what experiments are performed, this reconstruction provides an interpretation of the quantum probabilities that is as close as it can be to the classical probabilities. In addition, it is equivalent to QFT. Admittedly, this summary contains too much information packaged in a small part of an article that does not deal with probabilities, but all of these are described in more detail in {\citep{Stoica2023TheRelationWavefunction3DSpaceMWILocalBeablesProbabilities}} and the more pedagogical introduction {\citep{Stoica2024ClassicalManyWorldsInterpretation}}.
\end{remark}

\section{Symmetry Group Explanation}
\label{s:explanation_symmetry}

In the previous section, we saw that the wavefunction does not have qualitatively new features that were not already encountered in Euclidean geometry or in classical physics. The difference is at best quantitative, including all configurations of the points (see Remark \ref{rem:no-qualitative-difference}).

By appealing to symmetry groups, this section provides a deeper understanding of what was presented in Section \ref{s:explanation_Euclidean}. The geometric nature of the wavefunction is better understood in terms of spacetime symmetries, and, as shown by Wigner and Bargmann, this approach leads straightforwardly to the classification of particles based on spin and mass. Additional properties of particles, like electric charge and color, come from the local gauge symmetries, which also account for interactions. I will explain how Wigner's approach is a natural application of Klein's Erlangen Program for geometry. As such, spacetime and gauge symmetries are not simply properties of quantum theory, but they determine its very structure, including the properties of the wavefunction.

\subsection{In the Light of Klein Geometry}
\label{s:klein}

A homogeneous space is like a regular polyhedron; whatever two vertices of the polyhedron one may choose, the polyhedron can be rotated to bring the first vertex to the same place where the second vertex was so that, after this rotation, no observable difference exists. In the continuum limit, the polyhedron becomes a sphere, but there are other homogeneous spaces, like the Euclidean plane and space, or, in fact, Euclidean spaces of any dimension, and hyperbolic spaces. By dropping the measure of angles in a Euclidean space, and maintaining only the notion of parallelism, one obtains an affine space. In a similar way, we can obtain projective spaces and other spaces. All of these types of geometry have something in common; they, in some sense, ``look the same'' at each of their points. It turns out that they all can be described completely using symmetry groups.

In his 1872 Erlangen Program paper, Klein explained what various homogeneous geometries have in common \citep{klein1872ErlangenProgram}:
\begin{frem}
\begin{enumerate}
\item All that we can know about a homogeneous space $S$ is contained in its symmetry group $G$ and the action of $G$ on $S$ (the way in which $G$ transforms $S$).
\item The homogeneous space $S$ can be obtained by ``dividing'' its symmetry group $G$ by a subgroup $H$ of $G$. The result, the coset space $G/H$, can be identified with the homogeneous space $S$.
\end{enumerate}
\end{frem}

Homogeneous geometries include Euclidean, non-Euclidean, affine, projective, Minkowski spacetime geometries, etc., but not Riemannian geometry, which, in general, is not homogeneous (for which we can use a generalization of Klein's idea named Cartan geometry \citep{Sharpe2000DifferentialGeometryCartansGeneralizationOfKleinsErlangenProgram}).
Klein's major insight is that, at the core of each of the homogeneous geometries, is a group $G$ (in general, a Lie group) acting transitively, and, in general, effectively, on a space $S$. 
Let me explain these terms. A left action $\alpha$ of the group $G$ on the space $S$ is a way to associate to each element $g$ of $G$ a one-to-one function $\alpha_g:S\to S$ so that the following apply:
\begin{enumerate}
	\item 
If $e\in G$ is the identity element, $\alpha_e(s)=s$ for all $s\in S$;
	\item 
For any $g,g'\in G$, the composition $\alpha_g\circ\alpha_{g'}$ of $\alpha_g$ and $\alpha_{g'}$ satisfies $\alpha_g\circ\alpha_{g'}(s)=\alpha_{gg'}(s)$.
\end{enumerate}

In general, we will denote the left action of a group $G$ on a space $S$ by $(G,S,\cdot)$, and, instead of $\alpha_g(s)$, we will use the notation $g\cdot s$.

For a left action of a group $G$ on a space $S$, the orbit of an element $s\in S$ is the set of all elements $s'$ that can be obtained by acting with some element of the group $G$ on $s$. The left action of a group $G$ on a closed subgroup $H$ of $G$ also has orbits, and, in this particular case, the orbits are called left cosets.

The interested reader can find out more about homogeneous spaces in {(\citealp{KobayashiNomizu1996FoundationsOfDifferentialGeometryVol2}, Chapter X).

An action of $G$ on $S$ able to transform any point of space into any other point is called transitive. This makes $S$ into a homogeneous space for the group $G$. An action is called effective if the only identity transformation of $S$ by elements of the group $G$ is due to the identity of $G$.
The group $G$ is a transformation group for the space $S$.
Its action on the space $S$ is a representation of $G$ on $S$.

The geometric objects consist of subsets of the space $S$, 
for example, the points, lines, angles, planes, and so on can be geometric objects. Therefore, in general, the geometric objects are objects of Types \ref{type:space_subset} or \ref{type:space_collection_subsets}.
In Klein's view, geometry studies the invariant properties of geometric objects, i.e., those properties that remain unchanged under the action of the symmetry group. 

Two objects related by a symmetry transformation are congruent, or, in general, isomorphic. The notion of isomorphism generalizes the notions of congruence and isometry, which, in Euclidean geometry, are established by translations and rotation, to other geometries, including those that do not have notions like distance or angle, i.e., without a scalar product, or which have different kinds of scalar products. For example, affine geometry does not have a scalar product. Minkowski geometry has a scalar product, but this has a different signature from that in Euclidean geometry.

The set of all geometric objects of a given kind forms itself a space on which the group acts.
In Euclidean geometry, for example, a triangle can be transformed into another one by translations, rotations, and reflections. 
Such a transformation exists only if the two triangles are congruent, so this time the group action is not transitive.
However, the space of triangles can be decomposed into orbits, where each orbit consists of the set of all triangles congruent to a particular triangle.
The group action is transitive on each orbit.
Each element of the group transforms a point of that space into another one, corresponding to a transformation of an object in that space into another one.

More technical definitions and results from group theory and Klein geometry, as well as the discussion of the geometry of Minkowski spacetime as a Klein geometry, can be found in Appendix {\ref{appendix:klein}}.

Let us analyze the objects of Types \ref{type:space_subset}--\ref{type:space_additional_properties} from Section \ref{s:types} from the point of view of Klein's approach to geometry.

\begin{remark}
\label{rem:klein:type:space_subset}
According to Klein, the transformation group also acts on manifolds whose points are configurations of points or geometric objects from the original space \citep{klein1872ErlangenProgram}:
\begin{fquote}
We may use instead of the point any configuration contained in the manifoldness, -- a group of points, a curve or surface, \etc. As there is nothing at all determined at the outset about the number of
arbitrary parameters upon which these configurations shall depend, the number of dimensions of our line, plane, space, \etc., may be anything we like, according to our choice of the element.
\end{fquote}

These geometric objects are Type \ref{type:space_subset} objects.
For example, the set $\mc{T}$ of all triangles congruent to a general fixed triangle $\triangle ABC$ in the plane form a $3$-dimensional manifold, even though the plane itself is $2$-dimensional. The reason is that we can parametrize these triangles by their center of mass and the rotation angle. Since these triangles are all congruent, the group of isometries of the Euclidean plane acts transitively on $\mc{T}$.
\end{remark}

\begin{remark}
\label{rem:klein:principal_group}
Klein points out that the essential structure defining a geometry is the group structure (emphasis in the original text) \citep{klein1872ErlangenProgram}:
\begin{fquote}
\emph{But as long as we base our geometrical investigation upon the same group of transformations, the substance of the geometry remains unchanged.} [...]
The essential thing is, then, the group of transformations; the number of dimensions to be assigned to a manifoldness appears of secondary importance.
\end{fquote}

But the action of the Euclidean group on the set $\mc{T}$ of triangles in the plane is determined by its action on the plane $\R^2$, which provides the standard or defining representation of the Euclidean group. Therefore, $\R^2$ retains its fundamental character despite the existence of other associated representations in higher dimensions. Klein gives as examples more types of manifolds associated with space, in particular those consisting of lines, planes, curves, or surfaces, following ideas introduced by Pl\"ucker and Grassmann.
\end{remark}

\begin{remark}
\label{rem:klein:type:space_collection_subsets}
Klein also considers composite objects. For example, he wrote the following \citep{klein1872ErlangenProgram}:
\begin{fquote}
We must evidently choose our space-elements in such a way that their manifoldness either is itself a body or can be decomposed into bodies.
\end{fquote}

Therefore, Type \ref{type:space_collection_subsets} objects are geometric objects in space in Klein's geometries too.
\end{remark}

\begin{remark}
\label{rem:klein:type:space_additional_properties_geometrization}
Other types of manifolds, on which the principal group acts, can be obtained as the product between the original space and another space. This is the case for fiber bundles.
But are fiber bundles over space or spacetime present in Klein geometry?
Indeed they are; the map that associates to each element of the group $G$ its coset $G\mapsto G/H$ is a fiber bundle. The left cosets of the form $g\cdot H$, where $g\in G$, are the fibers, $S=G/H$ is the base space, and $G$, being the union of all of the left cosets, is the total space of the fiber bundle.
More general geometries obtained from the objects of a Klein geometry are also fiber bundles. Recall the example of triangles in Euclidean geometry mentioned above. The space of triangles is the total space, the orbits consisting of congruent triangles form the fibers, and the ``free-floating'' triangles of given side lengths form the base space.
In general, since fiber bundles are themselves manifolds on which the isometry group of the base space acts, Type \ref{type:space_additional_properties_geometrization} objects qualify as objects on space as well, according to Klein's description of ``manifoldness''.
\end{remark}

Let us summarize these remarks.
\begin{frem}
\begin{observation}
\label{obs:types_are_klein_geometry}
Types \emph{\ref{type:space_subset}}, \emph{\ref{type:space_collection_subsets}}, and \emph{\ref{type:space_additional_properties_geometrization}} are geometric objects in Klein geometry, as described in \citep{klein1872ErlangenProgram}.
In particular, by Theorem \emph{\ref{thm:wf_space}}, the wavefunction is an object of Klein geometry.
\end{observation}
\end{frem}
\begin{proof}
Follows from Remarks \ref{rem:klein:type:space_subset}--\ref{rem:klein:type:space_additional_properties_geometrization}.
\end{proof}

These results show that, according to Klein geometry, the space on which the wavefunction is an object is our three-dimensional space. While the $3\mathbf{n}${-dimensional}
 configuration space is very useful in NRQM, and the wavefunction is defined and propagates in this space, the wavefunction is also an object of the geometry of the three-dimensional space, according to Klein's characterization of geometries.
Moreover, while geometric objects can be understood as belonging to more different geometries, we will see that the natural home for the wavefunction is the geometry of spacetime.
In the following, we will explore the evidence for this based on the fact that the theory of the representations of the symmetry group of the Minkowski spacetime gives the particles with the values of spin that we observed, as shown by Wigner and Bargmann.

\subsection{Wigner's Theorem}
\label{s:wigner}

In the following, we will see more precisely how the wavefunction is understood based on Wigner's classification, and how this fits into Klein's Erlangen Program. 

Physical systems can be subject to symmetry transformations of space or spacetime, of the internal spaces of the gauge groups, and of permutations of identical particles.
Since the quantum states of physical systems are represented by vectors in a Hilbert space $\hilbert$, these symmetries induce transformations of the Hilbert space $\hilbert$.

A quantum system is not uniquely represented by a state vector $\ket{\psi}\in\hilbert$, but by a ray in $\hilbert$. Therefore, the actions of the symmetry groups are in fact representations on the projective Hilbert space defined by the rays of $\hilbert$, $P\(\hilbert\):=\(\hilbert\setminus\{0\}\)/\sim$, where $\ket{\psi}\sim\ket{\psi'}$ if there is some $\lambda\in\C\setminus\{0\}$ so that $\ket{\psi}=\lambda\ket{\psi'}$.
Let us denote by $\rayket{\psi}:=\ket{\psi}/\sim\in P\(\hilbert\)$ the ray corresponding to a state vector $\ket{\psi}\in\hilbert$, and by $\raybra{\psi}=\rayket{\psi}^\dagger\in P(\hilbert^\ast)$ its dual, which is the ray of $\bra{\psi}\in\hilbert^*$.

Projective transformations preserve the ray scalar product, defined as
\begin{equation}
\label{eq:ray_product}
\raybraket{\psi_1}{\psi_2}:=\abs{\braket{\psi_1}{\psi_2}}.
\end{equation}

We are interested in the representations of a symmetry group $G$ on the projective Hilbert space $P\(\hilbert\)$.
Let $T_g$ be the symmetry transformation of $P\(\hilbert\)$ corresponding to the element $g\in G$. Since the action of each transformation $T_g$ has to preserve the transition probabilities, it follows that, for any $\rayket{\psi_1},\rayket{\psi_2}\in P\(\hilbert\)$,
\begin{equation}
\label{eq:preserve_transition_probabilities}
\raybraket{T_g\psi_1}{T_g\psi_2}=\raybraket{\psi_1}{\psi_2},
\end{equation}
where $\raybra{T_g\psi}:=\(T_g\rayket{\psi}\)^\dagger$.
 
Wigner's theorem shows that, for $\dim\hilbert\geq 2$, any transformation $T_g$ of $P\(\hilbert\)$ satisfying \eqref{eq:preserve_transition_probabilities} is induced by a transformation $\wh{T}_g$ of $\hilbert$ which is either unitary or \mbox{antiunitary \citep{Wigner1931GruppentheorieUndIhreAnwendungAufDieQuantenMechanikDerAtomspektren,Wigner1959GroupTheoryAndItsApplicationToTheQuantumMechanicsOfAtomicSpectra}.}
If $\dim\hilbert=1$, there is a unitary transformation and an antiunitary transformation.
A transformation $\wh{T}$ of a Hilbert space $\hilbert$ is antiunitary if $\braket{\wh{T}\psi_1}{\wh{T}\psi_2}=\braket{\psi_1}{\psi_2}^\ast$. An antiunitary transformation $\wh{T}$ is additive, i.e., $\wh{T}\(\ket{\psi_1}+\ket{\psi_2}\)=\wh{T}\ket{\psi_1}+\wh{T}\ket{\psi_2}$, and antilinear, i.e., $\wh{T}\(\lambda\ket{\psi_1}\)=\lambda^\ast\wh{T}\ket{\psi_1}$, for all $\ket{\psi_1},\ket{\psi_2}\in\hilbert$.

Given the projective transformation $T_g$, the corresponding unitary or antiunitary transformation $\wh{T}_g$ is unique up to a phase factor $e^{i\vartheta}$ for each $g\in G$. Since the product of two transformations, $\wh{T}_g$ and $\wh{T}_{g'}$, satisfies $\wh{T}_{gg'}=e^{i\vartheta(g,g')}\wh{T}_g \wh{T}_{g'}$ for some $\vartheta(g,g')\in\R$, the freedom given by $\vartheta$ is useful for ensuring that the projective representation corresponds to a representation of $G$ on $\hilbert$. 
This is not always possible, but Wigner showed that it is possible locally around the identity \citep{Wigner1939OnUnitaryRepresentationsOfTheInhomogeneousLorentzGroup}.

Bargmann showed that representations that induce the projective representations do not exist for the rotation group $\SO(3)$, the Lorentz group, and the Poincar\'e group, but that they do exist for the universal covers of these groups, and they induce the projective representations of these groups \citep{Bargmann1964NoteOnWignerTheoremSymmetryOperations}. The universal cover group of the proper orthochronous Lorentz group $\SO^\uparrow(1,3)$---which is the connected component of the Lorentz group containing the identity---is its double cover $\Spin(1,3)\equiv\SL(2\C)$ (which is the spin group). The full Lorentz group is recovered by combining it with the time reversal operator $T$ and the parity operator $P$, which generate a group $\{1,P,T,PT\}$ isomorphic to $\Z_2\times\Z_2$ (incidentally named the Klein four-group). If we denote by $\wt{\mc{P}}_0:=\R^{1,3}\rtimes \Spin(1,3)$ the universal cover of the connected component of the Poincar\'e group containing the identity, the universal cover of the full Poincar\'e group is $\wt{\mc{P}}:=\wt{\mc{P}}_0\times\Z_2\times\Z_2$.
The fact that, in this case, the universal cover is the double cover explains why the wavefunction of a spin $\frac 12$ particle changes its sign under a $2\pi$ rotation.

Thus, Wigner and Bargmann showed that the projective representations of the Poincar\'e group on a Hilbert space $\hilbert$ correspond to representations of its universal cover by unitary and antiunitary transformations on $\hilbert$, and that particles, both fermions and bosons, are classified by spin and mass. Moreover, the evolution equation for free quantum states follows automatically from the symmetry according to the spin and mass. In particular, for massive free spin $\frac 1 2$ particles, they recovered the Dirac equation.

Other properties of quantum particles, like electric charge, color, and weak isospin, require the existence of fiber bundles, whose fibers are representations of other groups like $\U(1)$, $\SU(2)$, and $\SU(3)$. The idea of symmetry is related in their case to the Maxwell and Yang--Mills equations, and the charges result as invariants of the representations. All these groups and the group  $\wt{\mc{P}}_0$ act on the full Hilbert space, and their actions commute with each other. This means that the Wigner classification and the classification due to gauge theory complement each other in giving full classifications of the particles of the Standard Model based on representations.

It is not evident that merely making use of symmetry groups and their representations in the way Wigner and Bargmann did actually is Klein geometry. The following Theorem clarifies the relation between the two.

\begin{frem}
\begin{theorem}
\label{thm:wigner_klein}
The particle wavefunctions in the Wigner classification are geometric objects on spacetime, as prescribed in Klein geometry.
\end{theorem}
\end{frem}
\begin{proof}
We apply the Orbit--Stabilizer Theorem (Theorem \ref{thm:orbit_stabilizer} from Appendix \ref{appendix:klein}) for the Wigner representation of the universal cover of the Poincar\'e group $\wt{\mc{P}}$ on the Hilbert space $\hilbert$. This action is not transitive, but, from Theorem \ref{thm:orbit_stabilizer}, all of its orbits are classified by the cosets in $\wt{\mc{P}}$ of the form $\stab(\rayket{\psi})$, where $\rayket{\psi}\in \hilbert$.

The representations of $\wt{\mc{P}}$ on $\hilbert$ are ultimately ``made of'' homogeneous spaces, which are identifiable with the orbits of the action of $\wt{\mc{P}}$, and, by Theorem \ref{thm:orbit_stabilizer}, with cosets of $\wt{\mc{P}}$. The Wigner representation is decomposable in an invariant way into orbits, which are homogeneous spaces for $\wt{\mc{P}}$ and therefore are isomorphic to coset spaces of $\wt{\mc{P}}$. By Remark \ref{rem:klein:type:space_collection_subsets}, they are objects of Klein geometry, so the particle wavefunctions in the Wigner classification are geometric objects on spacetime.
\end{proof}

From Theorem \ref{thm:wigner_klein}, the Hilbert space $\hilbert$ is decomposable into infinitely many orbits, all of which are of a dimension that is at most the dimension of the group $\wt{\mc{P}}$, which is $10$
. Since $\hilbert$ is infinite dimensional, and the coset spaces of $\wt{\mc{P}}$ are finite dimensional, it follows that $\hilbert$ is decomposable into infinitely many orbits. Each part of the wavefunction from an orbit is transformed ``rigidly'' (isometrically) by $\wt{\mc{P}}$ within its own orbit, just like rigid geometric objects transform in spacetime. At the same time, even if they transform in the same way, the structure of each part of the wavefunction is more complicated than that of subsets in the Minkowski spacetime. They cannot be interpreted as such subsets, but they can be understood as geometric objects of Type \ref{type:space_additional_properties_geometrization}, as in Theorem \ref{thm:wf_space}, Corollary \ref{thm:wf_space_qft}, and Theorem \ref{thm:qf_space}.

Theorem \ref{thm:wigner_klein} applies even when gauge fields and the corresponding charges are present, because the universal cover $\wt{\mc{P}}_0$, by commuting with the gauge groups, acts on the Hilbert space in the same way as $\wt{\mc{P}}$ in Theorem \ref{thm:wigner_klein}.
Here, we should take $\wt{\mc{P}}_0$ rather than the full group $\wt{\mc{P}}$ because the Standard Model is not invariant under the $P$, $T$, or $C$ (charge conjugation) transformations individually, but only under the combined $CPT$ transformation.

\section{Consistency with Experiments}
\label{s:empirical}

The consistency of the wavefunction being an object on space with experiments has been empirically verified countless times, implicitly. The following examples are intended to explain how experiments can be understood as being about wavefunctions as objects on space.

\begin{experiment}
\label{exp:rigid}
Consider any macroscopic rigid object, for example, a rock.
The only transformations it can undergo without breaking or changing its structure are spacetime isometries, i.e., Galilei transformations in the nonrelativistic case, and Poincar\'e transformations in the relativistic case. The rock is a rigid object, and its symmetries are reduced to isometries of spacetime. We perform this experiment implicitly numerous times every day, whenever we manipulate rigid objects. Yet the rock is made of atoms, and, in each atom, electrons and nucleons are entangled. Even in the Helium atom, at the ground-state energy level, the two electrons are entangled. In fact, already in Hydrogen there is entanglement between the electron and the proton, which can be seen in a more refined model than the standard {\schrod} solution in which the proton has infinite inertia and no size \citep{Tommasini1998HydrogenEntangledSystem}. This is even more the case for heavier atoms. Moreover, molecular and crystalline structures involve entanglement between the atoms. This is not macroscopic superposition as in the {\schrod} cat experiment, and it is not likely that it can be harnessed for quantum computing, but it clearly shows that even a rock is a quantum system whose wavefunction requires a high-dimensional configuration space.

However, the possible transformations we can apply to the rock without changing its structure are very limited.
While of course we can imagine all sorts of unitary transformations of the Hilbert space of the rock, the ones leaving its spatial structure unchanged are only phase changes, permutations of identical particles, gauge transformations, and the spacetime isometries. From these, the only ones we can apply to the system without changing its structure, and whose effects are observable, are the spacetime isometries, under which the rock transforms like a classical rigid object, as in Theorem \ref{thm:wigner_klein}.

One may object that, with all the entanglement already present in the atoms and molecules of a rock, this is still limited to a very small subset of the Hilbert space necessary to describe it. And, indeed, the total Hilbert space contains infinitely many more {\schrod}'s-cat-like states that we do not observe in reality. But even quantum superpositions of different ``rock states'' transform like rigid objects, because spacetime isometries apply simultaneously to each term in the superposition.
\end{experiment}

\begin{experiment}
\label{exp:quantum-computer}
Consider a quantum computer, or a laboratory in which a quantum experiment is performed. Let us assume it to be mobile, for example, that it is placed on a ship moving with constant velocity. Due to the principle of relativity, the experiment is not affected by this motion. In general, it is not affected by an isometry of spacetime, and this holds for both a quantum laboratory and a quantum computer. But, apart from transformations representing spacetime isometries, all other unitary or antiunitary transformations are either not observable or they affect the structure of the experiment or the computation.
Now, it is clearly true that interactions are not captured in the Wigner--Bargmann approach, which is based solely on the Poincar\'e group. But they are accounted for by gauge theory, and local gauge transformations commute with the spacetime isometries, which is why, even when interactions are present, the wavefunction is still an object on space.
\end{experiment}

\begin{experiment}
\label{exp:prepare-and-measure}
When we perform experiments, we manipulate the measuring apparatus and the system on which we experiment. This manipulation happens in space. The human experimentalists and the measuring devices, as well as the apparatuses utilized to prepare the observed system, are manipulated in space. Even when we perform quantum mechanics experiments, the preparation and the measurement take place in space. The observed system may consist of more subsystems in entanglement, and the experiment may involve spacelike separated measurements, like in the EPR experiment \citep{EPR35,Bohm1951TheParadoxOfEinsteinRosenAndPodolsky}. But even in this case, we measure the state of the observed system in space. Experimental protocols include avoiding the contamination of the experiment with other systems, for example, when we detect interference or measure the spin of an atom, we have to make sure that the particle or atom is the intended one and not an intruder. The reason is that, if more particles arrive at the same region of space, the detection can be compromised.
The detector is unable to distinguish which copy of the three-dimensional space from the configuration space is the one corresponding to the observed particle, even though the wavefunction of all particles is interpreted as propagating on the configuration space. This means that, ultimately, the detection confirms that the particle occupies the same space as the detector, and not a different copy of space in the configuration space. This is true, in general, for all interactions between particles.
\end{experiment}

\section{Discussion}
\label{s:discussion}

We have seen that there are several ways to regard the wavefunction as being completely equivalent, from the point of view of the data it encodes, to geometric structures or fields that, in geometry and in classical physics, are considered objects on space.

\begin{remark}
\label{rem:fiber-bundle-nonlocal}
Some of these used fields are sections of fiber bundles. Even if the fields and connections of the fiber bundles are local objects, they can have nonlocal properties.
An example is the case of the M\"obius strip, which locally looks like a rectangle $(a,b)\times \R$ but, globally, has a twist, which makes it into a non-orientable surface. Even so, the sections of the M\"obius strip, which are just like the lines that one may draw when trying to draw a path around the loop, are fields on the base curve used to construct the M\"obius strip.

Other kinds of nonlocal properties can arise from such local structures because of the connection. A connection can be used to realize the parallel transport of a vector or tensor from a fiber along a curve. In the case of electromagnetism, this leads to an interesting nonlocal effect called the Aharonov--Bohm effect {\citep{AharonovBohm59SignificanceOfElectromagneticPotentialsInQuantumTheory}}.

This emergence of nonlocal properties from local geometric objects was used to represent the wavefunction as a local infinite-dimensional field on space in {\citep{Stoica2019RepresentationOfWavefunctionOn3D}}. This field contains all the information encoded in the wavefunction. However, for the wavefunction, being defined on the configuration space, when we try to understand it on space, it seems that the parts are unable to contain all the information about the whole wavefunction.
Is this holistic property shared by the infinite-dimensional field used in {\citep{Stoica2019RepresentationOfWavefunctionOn3D}}?
Indeed, it is shared, as was shown in \S III B in {\citep{Stoica2019RepresentationOfWavefunctionOn3D}}. We can split the space into regions, and split the infinite-dimensional field into its restrictions to these regions, and yet we can still recover the entire wavefunction. For this, a connection is used, which glues the local fiber bundles over each region back, recovering the fiber bundle over the entire space.
The more direct constructions presented in this article, which are based on sections of fiber bundles, share the same feature because it is a feature of the bundles. We will not even need to define a sophisticated connection as in {\citep{Stoica2019RepresentationOfWavefunctionOn3D}} because the fiber bundles used here are products, and are endowed with a natural connection. In the same way, we do not need to be explicit about the connection when derivating fields on a flat space as in Newtonian mechanics.
\end{remark}

The purpose of this research was to show why, from a geometric point of view, the wavefunction is an object on space. The purpose was for ontological reasons only, because, in most applications, the wavefunction on the configuration space is easier to use, even if only because the bulk of the literature and developments use it.

Some authors claim that the fact that the wavefunction is, usually and most easily, represented as defined and as propagating on the configuration space means that quantum systems cannot be on the three-dimensional space, and, in particular, objects like chairs and tables cannot make sense in a theory in which only the wavefunction exists and is governed by the {\schrod} dynamics \citep{Bell2004QuantumMechanicsForCosmologists,Maudlin2010CanTheWorldBeOnlyWavefunction,Maudlin2019PhilosophyofPhysicsQuantumTheory,Norsen2017FoundationsQM}.
This argument is often used to claim that the wavefunction needs to be supplemented either by point-particles, as in the pilot-wave theory \citep{Bohm1952SuggestedInterpretationOfQuantumMechanicsInTermsOfHiddenVariables}, or the positions of the spontaneous localization as in the GRW proposal \citep{GhirardiRiminiWeber1986GRWInterpretation} with the flash ontology (GRWf) \citep{Bell2004SpeakableFlashOntology}. Furthermore, it is argued that, by failing to provide such spatial features, wavefunction-only interpretations like Everett's interpretation \citep{Everett1957RelativeStateFormulationOfQuantumMechanics,Everett1973TheTheoryOfTheUniversalWaveFunction,Wallace2012TheEmergentMultiverseQuantumTheoryEverettInterpretation} and even the GRW interpretation with the mass-density ontology (GRWm) are insufficient.

Let us consider first how the pilot-wave theory is supposed to address the problem of how objects like chairs and tables are localized in space. The claim is that the point-particles do this job naturally, since they are located in space. But it is not sufficient for the point-particles to be located in a configuration of positions in space that look like a chair---the chair configuration should remain a chair configuration after a reasonable amount of time has passed, assuming nothing violent happens in the meantime. And since the point-particles are guided by the wavefunction, and since the probability of the point-particle configuration being located at a certain set of positions is determined by the wavefunction's amplitude at the point in the configuration space corresponding to that set of positions, the wavefunction itself should retain a stable ``chair-like'' configuration. So, to solve the problem for point-particles, we should solve it first for the wavefunction itself, and the pilot-wave theory turns out to have the same problem, attributed by these authors, for instance, to Everett's interpretation, which has to appeal to decoherence. In fact, Bohm realized this, and anticipated the necessity of decoherence for the pilot-wave theory \citep{Bohm1952SuggestedInterpretationOfQuantumMechanicsInTermsOfHiddenVariables}.

In the case of the GRWf interpretation, decoherence seems to be unnecessary. But the probability that the collapse localizes the wavefunction at a certain point in the configuration space, corresponding closely to a chair-like configuration of points in space, also depends on the wavefunction's amplitude at that point in the configuration space. To ensure that successive spontaneous localizations, and therefore ``flashes'', happen in such a way that the chair is stable, the wavefunction itself must be ``chair-like'', and the spontaneous collapse should be sharp, but not too sharp to sabotage this. Banning the wavefunction from the ontology of the GRW interpretation and retaining only the flash ontology does not change the need for the wavefunction to be ``chair-like''.

So, if the problem of the manifest image exists for standard quantum mechanics or for Everett's interpretation, it exists as well for the pilot-wave theory and GRW theory, regardless of the additional elements like point-particles or positions of the flashes.

The pilot-wave theory's point-particles and the GRWf flashes are called local ``beables'' by Bell \citep{Bell2004SpeakableUnspeakable}.
These are the elements of the theory that are localized in space.
Can standard quantum mechanics have local beables?
Bell himself considers that they have---and he gives as an example “the settings of switches and knobs and currents needed to prepare the initial unstable nucleus” \citep{Bell2004TheTheoryOfLocalBeables}.
But these are not the only local beables.
Suppose we measure a complete set of commuting observables of an elementary particle and find definite values. The wavefunction of the particle is then separable for the rest of the universe, so it can be understood straightforwardly as an object on space.
But if we measure an incomplete set of observables, the wavefunction is not completely determined, and the particle can be entangled with other systems. Then, it will be described by the reduced density matrix, and it will not be separable. Yet something will still be localized, namely, the observed properties, so they can be understood as local beables. This also applies to the observation of composite systems, for example, atoms, since, even though they consist of entangled particles, the observed properties of the atom are localized.
And the same happens in Everett's interpretation, where each branch of the wavefunction is exactly like the wavefunction that remains after the projection in standard quantum mechanics.
Another kind of local beable is given by the reduced density matrix of the quantum fields at every point of space \citep{WallaceTimpson2010QMOnSpacetime}.
All these kinds of local beables are just as good as the point-particles and the flashes.

But is it possible that all the information contained in a generic wavefunction can be contained in a set of local beables?
Even if the answer were negative, the same problem exists in all interpretations of quantum mechanics, including the pilot-wave theory and the GRW interpretation, since they all include the wavefunction.

Fortunately for all these interpretations, in the case that the manifest image really requires that the underlying ontology is spatial, we have seen in this article that the wavefunction as it is is indeed an object on space, in the same sense in which geometric objects are objects on space, both in Euclidean geometry and in Klein geometry, and that this is consistent with the empirical data.
Even so, the wavefunction is not explicitly a field on space, but why should we expect this? If this is important, it was previously shown that it can be faithfully represented as an infinite-dimensional field on space \citep{Stoica2019RepresentationOfWavefunctionOn3D}, and another way to see this was proposed in \citep{Stoica2023TheRelationWavefunction3DSpaceMWILocalBeablesProbabilities}. Even without such a representation, here, we made it clear that it already is an object on space, like classical objects or objects in Euclidean geometry or in Klein's Erlangen Program.

Does this imply that we can be handed a wavefunction and be able to draw on a piece of paper, based on it only, the tables and the chairs \citep{Albert2019How2TeachQM}? We can try to draw them, but the image will be very complicated, because, while this article shows that the wavefunction is qualitatively not different from other geometric objects like triangles, it is infinitely more complex. On the other hand, existing accounts based on the wavefunction being an object in the configuration space are able to explain very well the manifest image \citep{Wallace2012TheEmergentMultiverseQuantumTheoryEverettInterpretation}, in particular, if the fact that interactions are local in three-dimensional space is taken into account as in \citep{Vaidman2022WaveFunctionRealismAnd3Dimensions}. And the resulting picture is physically intuitive, without having to rely on the fact that the wavefunction is, in any sense, an object on space. However, proving a three-dimensional grounding of the wavefunction may help remove the anxiety of being lost in the configuration space that may be responsible for the feeling that the wavefunction is not enough to account for the physical world. The wavefunction provides a complete description of reality \citep{Vaidman2016AllIsPsi}.

%
%
%
%
%
%
%
%
%
%
%


\textbf{Acknowledgments.}
The author thanks Scott Aaronson, Eddy Keming Chen, and reviewers for their valuable comments and suggestions offered for a previous version of the manuscript. Nevertheless, the author bears full responsibility for the article.



\appendix

\section{Technical Details about Klein Geometry}
\label{appendix:klein}

In this Appendix, we will review the mathematical grounds of Klein geometry in a more precise way.
\begin{definition}[Orbits and stabilizers]
\label{def:orbit_stabilizer}
Let $(G,S,\cdot)$ be a left action of a group $G$ on a set $S$.
The orbit of a point $s\in S$ is the subset $\orb(s)\subseteq S$ of $S$ defined as
\begin{equation}
\label{eq:orb}
\orb(s):=\{g\cdot s\in S|g\in G\}.
\end{equation}

The stabilizer subgroup $\stab(s)$ (also isotropy group) of $G$ with respect to $s$ is defined as
\begin{equation}
\label{eq:stab}
\stab(s):=\{g\in G|g\cdot s=s\}.
\end{equation}
\end{definition}


\begin{frem}
\begin{theorem}[Orbit--Stabilizer Theorem. See, e.g., \citealp{Helgason1979DifferentialGeometryLieGroupsSymmetricSpaces}]
\label{thm:orbit_stabilizer}
Let $s\in S$. Then, there is a one-to-one map between the left cosets $G/\stab(s)$ and the elements of the orbit of $s$, defined by
\begin{equation}
\label{eq:orb-stab}
\begin{split}
&\varphi:G/\stab(s)\to\orb(s)\\
&\varphi\(g\stab(s)\)=g\cdot s.
\end{split}
\end{equation}

The map $\varphi$ defines an action isomorphism between the left action of $G$ on its left coset space and the action of $G$ on $S$. If $G$ is a Lie group, $\varphi$ is a diffeomorphism.
\end{theorem}
\end{frem}
\begin{proof}
See, for example, Theorem 3.2 (p. 121), Theorem 4.2 (p. 123), and Proposition 4.3 (p. 124) in \citep{Helgason1979DifferentialGeometryLieGroupsSymmetricSpaces}.
\end{proof}


\begin{definition}[Klein geometries]
\label{def:klein_geometries}
A Klein geometry $(G,H)$ consists of a Lie group $G$ and a topologically closed subgroup $H$ of $G$ \citep{Sharpe2000DifferentialGeometryCartansGeneralizationOfKleinsErlangenProgram}. The left coset space $G/H=\{gH|g\in G\}$ is identified as the space on which $G$ acts. This results in a principal fiber bundle $G\mapsto G/H$ with typical fiber and structure group $H$ (\citealp{Cohen1998TopologyOfFiberBundles}, Cor. 1.4).
\end{definition}

\begin{frem}
\begin{proposition}
\label{thm:homogeneous_space}
Any homogeneous space $S$ of a group $G$ can be obtained as the coset space $G/H$, where the subgroup $H$ of $G$ is the stabilizer of a point $s\in S$.
\end{proposition}
\end{frem}
\begin{proof}
Follows from Theorem \ref{thm:orbit_stabilizer}.
\end{proof}

\begin{example}
\label{ex:homogeneous_space}
For the Minkowski spacetime, $G$ is the Poincar\'e group $\R^{1,3}\rtimes\OO(1,3)$, where $\R^{1,3}$ is, in fact, the vector space $\R^4$, and $H$ is the Lorentz group $\OO(1,3)$, the group of transformations of the vector space $\R^{1,3}$ which preserves the symmetric bilinear form
\begin{equation}
\label{eq:poincare_matrix}
\eta=\begin{pmatrix}
1 & 0 & 0 & 0 \\
0 & -1 & 0 & 0 \\
0 & 0 & -1 & 0 \\
0 & 0 & 0 & -1 \\
\end{pmatrix}.
\end{equation}

An element of the Poincar\'e group has the form $(\mathbf{v},L)$, where $\mathbf{v}\in\R^{1,3}$ corresponds to a translation by a vector $\mathbf{v}$, and $L\in\OO(1,3)$.
The Poincar\'e group admits the matrix representation $\lambda:\R^{1,3}\rtimes\OO(1,3)\to\mathcal{M}(5,5)$,
\begin{equation}
\label{eq:poincare_matrix_block}
\lambda(\mathbf{v},L)=\begin{pmatrix}
1 & 0 \\
\mathbf{v} & L 
\end{pmatrix}.
\end{equation}

The element corresponding to the translation with a vector $\mathbf{v}$ has the form
$\lambda(\mathbf{v},I_4)=\begin{pmatrix}
1 & 0 \\
\mathbf{v} & I_4 
\end{pmatrix}$. The element corresponding to the Lorentz transformation $L$
has the form
$\lambda(\mathbf{0},L)=\begin{pmatrix}
1 & 0 \\
\mathbf{0} & L 
\end{pmatrix}$.

Let us denote the subgroup of translations and the subgroup of Lorentz transformations by
$\lambda(\R^{1,3},I_4):=
\begin{pmatrix}
1 & 0 \\
\R^{1,3} & I_4 
\end{pmatrix}$
and 
$\lambda(\mathbf{0},\OO(1,3)):=
\begin{pmatrix}
1 & 0 \\
\mathbf{0} & \OO(1,3) 
\end{pmatrix}$.
Then, the left cosets from $G/H$ in Definition \ref{def:klein_geometries} have the form $\lambda(\mathbf{v},L)\(\mathbf{0},\OO(1,3)\)$; therefore, they have the matrix form
\begin{equation}
\label{eq:poincare_cosets}
\begin{pmatrix}
1 & 0 \\
\mathbf{v} & L 
\end{pmatrix}
\begin{pmatrix}
1 & 0 \\
\mathbf{0} & \OO(1,3) 
\end{pmatrix}
=
\begin{pmatrix}
1 & 0 \\
\mathbf{v} & \OO(1,3) 
\end{pmatrix}.
\end{equation}

Since, in the RHS of Equation \eqref{eq:poincare_cosets}, the only free parameter is $\mathbf{v}$, the cosets $\begin{pmatrix}
1 & 0 \\
\mathbf{v} & \OO(1,3) 
\end{pmatrix}$ form a four-dimensional space identifiable with the Minkowski spacetime $\R^{1,3}$.

Similar interpretations as left coset spaces can be given for the Euclidean geometry on $\R^n$, and for the hyperbolic, elliptic, projective, affine, and conformal geometries.
\end{example}

By applying Definition \ref{def:klein_geometries} to the Poincar\'e group $\R^{1,3}\rtimes O(1,3)$, we obtain the fiber bundle $\R^{1,3}\rtimes O(1,3)\mapsto\R^{1,3}$, whose base space is the Minkowski spacetime, and whose typical fiber is diffeomorphic with the Lorentz group.
Since the Poincar\'e group acts on this bundle, it also transforms its sections and, in fact, all of its subsets, so they form representations.

\addcontentsline{toc}{section}{\refname}



\begin{thebibliography}{}

\bibitem[Aharonov and Bohm,
  1959]{AharonovBohm59SignificanceOfElectromagneticPotentialsInQuantumTheory}
Aharonov, Y. and Bohm, D. (1959).
\newblock Significance of electromagnetic potentials in quantum theory.
\newblock {\em Phys.\ Rev.}, 115(3):485.

\bibitem[Albert, 1996]{DavidAlbert1996ElementaryQuantumMetaphysics}
Albert, D. (1996).
\newblock Elementary quantum metaphysics.
\newblock In Cushing, J., Fine, A., and Goldstein, S., editors, {\em Bohmian
  mechanics and quantum theory: {A}n appraisal}, pages 277--284. Springer,
  Berlin/Heidelberg, Germany.

\bibitem[Albert, 2019]{Albert2019How2TeachQM}
Albert, D. (2019).
\newblock How to teach quantum mechanics.
\newblock {\em Preprint
  \href{http://philsci-archive.pitt.edu/15584/}{philsci-archive:00015584/}}.

\bibitem[Albert and Loewer, 1996]{AlbertLoewer1996TailsSchrodingerCat}
Albert, D. and Loewer, B. (1996).
\newblock Tails of {S}chr{\"o}dinger’s cat.
\newblock In {\em Perspectives on quantum reality}, pages 81--92. Springer,
  Berlin/Heidelberg, Germany.

\bibitem[Allori,
  2013]{Allori2013PrimitiveOntologyAndTheStructureOfFundamentalPhysicalTheories}
Allori, V. (2013).
\newblock Primitive ontology and the structure of fundamental physical
  theories.
\newblock In
  \cite{NeyAlbert2013TheWaveFunctionEssaysOnTheMetaphysicsOfQuantumMechanics},
  pages 58--75.

\bibitem[Allori et~al., 2008]{Allori2008CommonBMandGRW}
Allori, V., Goldstein, S., Tumulka, R., and Zangh{\`\i}, N. (2008).
\newblock On the common structure of {B}ohmian mechanics and the
  {G}hirardi--{R}imini--{W}eber theory.
\newblock {\em Br J Philos Sci}, 59(3):353--389.

\bibitem[Bacciagaluppi and Valentini,
  2009]{BacciagaluppiValentini2009SolvayConference}
Bacciagaluppi, G. and Valentini, A. (2009).
\newblock {\em Quantum theory at the crossroads: reconsidering the 1927
  {S}olvay {C}onference}.
\newblock Cambridge Univ. Press, Cambridge, UK.

\bibitem[Bargmann, 1964]{Bargmann1964NoteOnWignerTheoremSymmetryOperations}
Bargmann, V. (1964).
\newblock Note on {W}igner's theorem on symmetry operations.
\newblock {\em J. Math. Phys.}, 5(7):862--868.

\bibitem[Barrett, 1999]{Barrett1999TheQuantumMechanicsOfMindsAndWorlds}
Barrett, J. (1999).
\newblock {\em The quantum mechanics of minds and worlds}.
\newblock Oxford University Press, Oxford, UK.

\bibitem[Barrett, 2017]{Barrett2017TypicalWorlds}
Barrett, J. (2017).
\newblock Typical worlds.
\newblock {\em Stud. Hist. Philos. Mod. Phys.}, 58:31--40.

\bibitem[Bell, 1964]{Bell64BellTheorem}
Bell, J. (1964).
\newblock On the {E}instein-{P}odolsky-{R}osen paradox.
\newblock {\em Physics}, 1(3):195--200.

\bibitem[Bell, 2004a]{Bell2004SpeakableFlashOntology}
Bell, J. (2004a).
\newblock Are there quantum jumps?
\newblock In {\em Speakable and Unspeakable in Quantum Mechanics}, pages
  201--212. Cambridge Univ. Press, Cambridge, UK.

\bibitem[Bell, 2004b]{Bell2004LaNouvelleCuisine}
Bell, J. (2004b).
\newblock La nouvelle cuisine.
\newblock In \cite{Bell2004SpeakableUnspeakable}, pages 232--248.

\bibitem[Bell, 2004c]{Bell2004QuantumMechanicsForCosmologists}
Bell, J. (2004c).
\newblock Quantum mechanics for cosmologists.
\newblock In \cite{Bell2004SpeakableUnspeakable}, pages 117--138.

\bibitem[Bell, 2004d]{Bell2004SpeakableUnspeakable}
Bell, J. (2004d).
\newblock {\em Speakable and unspeakable in quantum mechanics: {C}ollected
  papers on quantum philosophy}.
\newblock Cambridge Univ. Press, Cambridge, UK.

\bibitem[Bell, 2004e]{Bell2004TheTheoryOfLocalBeables}
Bell, J. (2004e).
\newblock The theory of local beables.
\newblock In \cite{Bell2004SpeakableUnspeakable}, pages 52--62.

\bibitem[Belot, 2012]{Belot2012QuantumStatesPrimitiveOntologists}
Belot, G. (2012).
\newblock Quantum states for primitive ontologists.
\newblock {\em Eur. J. Philos. Sci.}, 2(1):67--83.

\bibitem[Bohm, 1951]{Bohm1951TheParadoxOfEinsteinRosenAndPodolsky}
Bohm, D. (1951).
\newblock {\em {T}he Paradox of {E}instein, {R}osen, and {P}odolsky}, pages
  611--623.
\newblock Prentice-Hall, Englewood Cliffs.

\bibitem[Bohm,
  1952]{Bohm1952SuggestedInterpretationOfQuantumMechanicsInTermsOfHiddenVariables}
Bohm, D. (1952).
\newblock {A} suggested interpretation of quantum mechanics in terms of
  ``hidden'' variables, {I \& II}.
\newblock {\em Phys. Rev.}, 85(2):166--193.

\bibitem[Bohm, 2004]{Bohm2004CausalityChanceModernPhysics}
Bohm, D. (2004).
\newblock {\em Causality and chance in modern physics}.
\newblock Routledge, London.

\bibitem[Brown and Wallace, 2005]{BrownWallace2005BohmVsEverett}
Brown, H. and Wallace, D. (2005).
\newblock Solving the measurement problem: De {B}roglie--{B}ohm loses out to
  {E}verett.
\newblock {\em Found. Phys.}, 35(4):517--540.

\bibitem[Chen, 2017]{EKChen2017OurFundamentalPhysicalSpace}
Chen, E. (2017).
\newblock Our fundamental physical space: {A}n essay on the metaphysics of the
  wave function.
\newblock {\em The Journal of Philosophy}, 114(7):333--365.

\bibitem[Cohen, 1998]{Cohen1998TopologyOfFiberBundles}
Cohen, R. (1998).
\newblock {\em The topology of fiber bundles. {L}ecture notes}.
\newblock Standford University.

\bibitem[Einstein et~al., 1935]{EPR35}
Einstein, A., Podolsky, B., and Rosen, N. (1935).
\newblock Can quantum-mechanical description of physical reality be considered
  complete?
\newblock {\em Phys. Rev.}, 47(10):777.

\bibitem[Emery, 2017]{Emery2017AgainstRadicalQuantumOntologies}
Emery, N. (2017).
\newblock Against radical quantum ontologies.
\newblock {\em Philosophy and Phenomenological Research}, 95(3):564--591.

\bibitem[Everett, 1957]{Everett1957RelativeStateFormulationOfQuantumMechanics}
Everett, H. (1957).
\newblock ``{R}elative state'' formulation of quantum mechanics.
\newblock {\em Rev. Mod. Phys.}, 29(3):454--462.

\bibitem[Everett, 1973]{Everett1973TheTheoryOfTheUniversalWaveFunction}
Everett, H. (1973).
\newblock {The Theory of the Universal Wave Function}.
\newblock In {\em {The Many-Worlds Hypothesis of Quantum Mechanics}}, pages
  3--137, Princeton, NJ. Princeton University Press.

\bibitem[Fine and Brown, 1988]{FineBrown1988ShakyGameEinsteinRealismQT}
Fine, A. and Brown, H. (1988).
\newblock The shaky game: {E}instein, realism and the quantum theory.
\newblock {\em Am. J. Phys.}, 56:571.

\bibitem[Floreanini and Jackiw,
  1988]{FloreaniniJackiw1988FunctionalRepresentationForFermionicQuantumFields}
Floreanini, R. and Jackiw, R. (1988).
\newblock Functional representation for fermionic quantum fields.
\newblock {\em Phys. Rev. D}, 37(8):2206.

\bibitem[Forrest, 1988]{Forrest1988QuantumMetaphysics}
Forrest, P. (1988).
\newblock {\em Quantum Metaphysics}.
\newblock Blackwell Pub, New York, NY, USA.

\bibitem[Gao, 2017]{Gao2017MeaningWavefunction}
Gao, S. (2017).
\newblock {\em The Meaning of the Wave Function: {I}n Search of the Ontology of
  Quantum Mechanics}.
\newblock Cambridge Univ. Press, Cambridge, UK.

\bibitem[Ghirardi et~al., 1986]{GhirardiRiminiWeber1986GRWInterpretation}
Ghirardi, G., Rimini, A., and Weber, T. (1986).
\newblock Unified dynamics of microscopic and macroscopic systems.
\newblock {\em Phys. Rev. D}, 34(2):470--491.

\bibitem[Goldstein and Zangh{\`i},
  2013]{GoldsteinZanghi2013RealityAndRoleOfWavefunction}
Goldstein, S. and Zangh{\`i}, N. (2013).
\newblock Reality and the role of the wave function in quantum theory.
\newblock In
  \cite{NeyAlbert2013TheWaveFunctionEssaysOnTheMetaphysicsOfQuantumMechanics},
  pages 91--109.

\bibitem[Guggenheimer, 1977]{Guggenheimer1977DifferentialGeometry}
Guggenheimer, H. (1977).
\newblock {\em Differential geometry}.
\newblock Dover Publications, Inc., New York.

\bibitem[Hatfield,
  2018]{Hatfield2018QuantumFieldTheoryOfPointParticlesAndStrings}
Hatfield, B. (2018).
\newblock {\em Quantum field theory of point particles and strings}.
\newblock CRC Press, Boca Raton, FL, USA.

\bibitem[Helgason,
  1979]{Helgason1979DifferentialGeometryLieGroupsSymmetricSpaces}
Helgason, S. (1979).
\newblock {\em Differential geometry, {L}ie groups, and symmetric spaces}.
\newblock Academic press, New York, NY, USA.

\bibitem[Howard, 1990]{Howard1990EinsteinWorriesQM}
Howard, D. (1990).
\newblock {"Nicht Sein Kann was Nicht Sein Darf,"} or the {Prehistory of EPR},
  1909--1935: {E}instein’s early worries about the quantum mechanics of
  composite systems.
\newblock In {\em Sixty-two years of uncertainty}, pages 61--111. Springer,
  Berlin/Heidelberg, Germany.

\bibitem[Jackiw,
  1988]{Jackiw1988AnalysisInfDimManifoldsSchrodingerRepresentationForQuantizedFields}
Jackiw, R. (1988).
\newblock Analysis on infinite-dimensional manifolds -- {S}chr{\"o}dinger
  representation for quantized fields.
\newblock In {\'E}boli, O., Gomes, M., and Santoro, A., editors, {\em Field
  theory and particle physics}, Singapore. World Scientific.

\bibitem[Klein, 1893]{klein1872ErlangenProgram}
Klein, F. (1893).
\newblock Vergleichende {B}etrachtungen {\"u}ber neuere geometrische
  {F}orschungen.
\newblock {\em Math. Ann.}, 43(1):63--100.

\bibitem[Kobayashi and Nomizu,
  1996]{KobayashiNomizu1996FoundationsOfDifferentialGeometryVol2}
Kobayashi, S. and Nomizu, K. (1996).
\newblock {\em Foundations of differential geometry. {V}olume 2}, volume~61.
\newblock John Wiley \& Sons, New York, London, Sydney.

\bibitem[Lewis, 2004]{Lewis2004LifeConfigurationSpace}
Lewis, P. (2004).
\newblock Life in configuration space.
\newblock {\em Br J Philos Sci}, 55(4):713--729.

\bibitem[Loewer, 1996]{Loewer1996HumeanSupervenience}
Loewer, B. (1996).
\newblock Humean supervenience.
\newblock {\em Philosophical Topics}, 24(1):101--127.

\bibitem[Maudlin, 2007]{Maudlin2007CompletenessSupervenienceOntology}
Maudlin, T. (2007).
\newblock Completeness, supervenience and ontology.
\newblock {\em J. Phys. A: Math. Theor.}, 40(12):3151.

\bibitem[Maudlin, 2010]{Maudlin2010CanTheWorldBeOnlyWavefunction}
Maudlin, T. (2010).
\newblock Can the world be only wavefunction.
\newblock In Saunders, S., Barrett, J., Kent, A., and Wallace, D., editors,
  {\em Many worlds?: {E}verett, quantum theory, \& reality}, pages 121--143.
  Oxford University Press, Oxford, UK.

\bibitem[Maudlin, 2013]{Maudlin2013TheNatureOfTheQuantumState}
Maudlin, T. (2013).
\newblock The nature of the quantum state.
\newblock In
  \cite{NeyAlbert2013TheWaveFunctionEssaysOnTheMetaphysicsOfQuantumMechanics},
  pages 184--202.

\bibitem[Maudlin, 2019]{Maudlin2019PhilosophyofPhysicsQuantumTheory}
Maudlin, T. (2019).
\newblock {\em Philosophy of physics: {Q}uantum {T}heory}, volume~33 of {\em
  Princeton Foundations of Contemporary Philosophy}.
\newblock Princeton University Press, Princeton, NJ.

\bibitem[Monton, 2002]{Monton2002WavefunctionOntology}
Monton, B. (2002).
\newblock Wave function ontology.
\newblock {\em Synthese}, 130(2):265--277.

\bibitem[Monton, 2006]{Monton2006QM3Nspace}
Monton, B. (2006).
\newblock Quantum mechanics and 3{N}-dimensional space.
\newblock {\em Philosophy of science}, 73(5):778--789.

\bibitem[Monton, 2013]{Monton2013Against3NSpace}
Monton, B. (2013).
\newblock Against 3{N}-dimensional space.
\newblock In
  \cite{NeyAlbert2013TheWaveFunctionEssaysOnTheMetaphysicsOfQuantumMechanics}.

\bibitem[Ney,
  2012]{Ney2012TheStatusOfOurOrdinaryThreeDimensionsInAQuantumUniverse}
Ney, A. (2012).
\newblock The status of our ordinary three dimensions in a quantum universe.
\newblock {\em No{\^u}s}, 46(3):525--560.

\bibitem[Ney, 2013]{Ney2013OntologicalReductionWavefunctionOntology}
Ney, A. (2013).
\newblock Ontological reduction and the wave function ontology.
\newblock In
  \cite{NeyAlbert2013TheWaveFunctionEssaysOnTheMetaphysicsOfQuantumMechanics}.

\bibitem[Ney and Albert,
  2013]{NeyAlbert2013TheWaveFunctionEssaysOnTheMetaphysicsOfQuantumMechanics}
Ney, A. and Albert, D. (2013).
\newblock The wave function: {E}ssays on the metaphysics of quantum mechanics.

\bibitem[Norsen, 2017]{Norsen2017FoundationsQM}
Norsen, T. (2017).
\newblock {\em Foundations of Quantum Mechanics}.
\newblock Springer, Switzerland, Cham, Switzerland.

\bibitem[North, 2013]{North2013StructureOfQuantumWorld}
North, J. (2013).
\newblock The structure of a quantum world.
\newblock In
  \cite{NeyAlbert2013TheWaveFunctionEssaysOnTheMetaphysicsOfQuantumMechanics},
  pages 184--202.

\bibitem[{Przibram, K. (ed)},
  1967]{Przibram1967LettersWaveMechanics}
{Przibram, K. (ed)} and {Klein, M.J. (trans)} (1967).
\newblock {\em Letters on Wave Mechanics: {S}chr{\"o}dinger, {P}lank,
  {E}instein, {L}orentz}.
\newblock Philosophical Library, New York.

\bibitem[{Przibram, K. (ed)},
  2011]{Przibram2011LettersWaveMechanics}
{Przibram, K. (ed)} and {Klein, M.J. (trans)} (2011).
\newblock {\em Letters on Wave Mechanics: {C}orrespondence with {H.A. Lorentz,
  Max Planck, and Erwin Schr{\"o}dinger}}.
\newblock Open Road Integrated Media, New York.

\bibitem[Pusey et~al., 2012]{PBR2012RealityOfPsi}
Pusey, M., Barrett, J., and Rudolph, T. (2012).
\newblock On the reality of the quantum state.
\newblock {\em Nature Phys.}, 8(6):475--478.

\bibitem[Schr{\"o}dinger,
  1926]{Schrodinger1926QuantisierungAlsEigenwertproblem}
Schr{\"o}dinger, E. (1926).
\newblock Quantisierung als {Eigenwertproblem}.
\newblock {\em Ann. Phys.}, 385(13):437--490.

\bibitem[Schr{\"o}dinger, 1982]{Schrodinger1982CollectedPapersWaveMechanics}
Schr{\"o}dinger, E. (1982).
\newblock {\em Collected papers on wave mechanics}, volume 302.
\newblock American Mathematical Soc., Providence, Rhode Island, USA.

\bibitem[Sharpe,
  2000]{Sharpe2000DifferentialGeometryCartansGeneralizationOfKleinsErlangenProgram}
Sharpe, R. (2000).
\newblock {\em Differential geometry: {C}artan's generalization of {K}lein's
  {E}rlangen program}, volume 166.
\newblock Springer Science \& Business Media, Berlin Heidelberg.

\bibitem[Stoica, 2019]{Stoica2019RepresentationOfWavefunctionOn3D}
Stoica, O. (2019).
\newblock Representation of the wave function on the three-dimensional space.
\newblock {\em Phys. Rev. A}, 100:042115.

\bibitem[Stoica, 2021]{Stoica2020StandardQuantumMechanicsWithoutObservers}
Stoica, O. (2021).
\newblock Standard quantum mechanics without observers.
\newblock {\em Phys. Rev. A}, 103(3):032219.

\bibitem[Stoica,
  2023]{Stoica2023TheRelationWavefunction3DSpaceMWILocalBeablesProbabilities}
Stoica, O. (2023).
\newblock The relation between wavefunction and {3D} space implies many worlds
  with local beables and probabilities.
\newblock {\em Quantum Reports}, 5(1):102--115.

\bibitem[Stoica, 2024]{Stoica2024ClassicalManyWorldsInterpretation}
Stoica, O. (2024).
\newblock Classical many-worlds interpretation.
\newblock {\em Preprint
  \href{https://arxiv.org/abs/2407.16774}{arXiv:2407.16774}}.

\bibitem[Tommasini et~al., 1998]{Tommasini1998HydrogenEntangledSystem}
Tommasini, P., Timmermans, E., and de~Toledo~Piza, A. (1998).
\newblock The hydrogen atom as an entangled electron--proton system.
\newblock {\em Am. J. Phys.}, 66(10):881--886.

\bibitem[Vaidman, 2016]{Vaidman2016AllIsPsi}
Vaidman, L. (2016).
\newblock All is $\psi$.
\newblock {\em J. Phys. Conf. Ser.}, 701:012020.

\bibitem[Vaidman, 2021]{SEP-Vaidman2021MWI}
Vaidman, L. (2021).
\newblock Many-worlds interpretation of quantum mechanics.
\newblock In Zalta, E., editor, {\em The {Stanford} Encyclopedia of
  Philosophy}. Metaphysics Research Lab, Stanford University, Stanford,
  California.
\newblock https://plato.stanford.edu/entries/qm-manyworlds/, last accessed
  \today.

\bibitem[Vaidman, 2022]{Vaidman2022WaveFunctionRealismAnd3Dimensions}
Vaidman, L. (2022).
\newblock Wave function realism and three dimensions.
\newblock In Allori, V., editor, {\em Quantum mechanics and fundamentality},
  volume 460, pages 195--2010, Berlin/Heidelberg, Germany. Springer Nature.

\bibitem[Wallace, 2002]{Wallace2002WorldsInMWI}
Wallace, D. (2002).
\newblock Worlds in the {E}verett interpretation.
\newblock {\em Stud. Hist. Philos. Mod. Phys.}, 33(4):637--661.

\bibitem[Wallace, 2003]{Wallace2003EverettAndStructure}
Wallace, D. (2003).
\newblock Everett and structure.
\newblock {\em Stud. Hist. Philos. Mod. Phys.}, 34(1):87--105.

\bibitem[Wallace,
  2012]{Wallace2012TheEmergentMultiverseQuantumTheoryEverettInterpretation}
Wallace, D. (2012).
\newblock {\em The emergent multiverse: {Q}uantum theory according to the
  {E}verett interpretation}.
\newblock Oxford University Press, Oxford, UK.

\bibitem[Wallace and Timpson, 2010]{WallaceTimpson2010QMOnSpacetime}
Wallace, D. and Timpson, C. (2010).
\newblock Quantum mechanics on spacetime {I}: {S}pacetime state realism.
\newblock {\em Br J Philos Sci}, 61(4):697--727.

\bibitem[Wigner,
  1931]{Wigner1931GruppentheorieUndIhreAnwendungAufDieQuantenMechanikDerAtomspektren}
Wigner, E. (1931).
\newblock {\em Gruppentheorie und ihre {A}nwendung auf die {Q}uanten mechanik
  der {A}tomspektren}.
\newblock Friedrich Vieweg und Sohn, Braunschweig, Germany.

\bibitem[Wigner,
  1939]{Wigner1939OnUnitaryRepresentationsOfTheInhomogeneousLorentzGroup}
Wigner, E. (1939).
\newblock On unitary representations of the inhomogeneous lorentz group.
\newblock {\em Ann. Math.}, pages 149--204.

\bibitem[Wigner,
  1959]{Wigner1959GroupTheoryAndItsApplicationToTheQuantumMechanicsOfAtomicSpectra}
Wigner, E. (1959).
\newblock {\em Group Theory and its Application to the Quantum Mechanics of
  Atomic Spectra}.
\newblock Academic Press, New York.

\bibitem[Yang and Mills,
  1954]{YangMills1954ConservationOfIsotopicSpinAndIsotopicGaugeInvariance}
Yang, C. and Mills, R. (1954).
\newblock Conservation of isotopic spin and isotopic gauge invariance.
\newblock {\em Phys. Rev.}, 96:191--195.

\end{thebibliography}
\end{document}